\documentclass[onecolumn,noversion,nonote]{cdcarticle}

\def\MODE{3}

\usepackage{amsmath,amsthm,amssymb}
\usepackage{tikz}
\usepackage{graphicx}
\usepackage{enumerate}
\usepackage{subcaption}
\usepackage[hidelinks]{hyperref}
\usepackage{array}
\usepackage{MnSymbol}
\usepackage{verbatim}
\usetikzlibrary{shapes}

\newtheorem{thm}{Theorem}
\newtheorem{problem}{Problem}
\newtheorem{lem}{Lemma}

\newtheorem*{prop*}{Proposition $P(s)$}
\newtheorem{rem}{Remark}
\newtheorem{defn}{Definition}
\newtheorem{cor}{Corollary}
\renewenvironment{proof}{{\noindent\bf Proof.}}{ \hfill ~\qed}

\def\qed{\rule[0pt]{5pt}{5pt}\par\medskip}

\newcommand{\tp}{\mathsf{T}}		
\newcommand{\Sbin}{\mathbb{S}}		
\newcommand{\eqt}{\buildrel\smash{\smash t}\over =}

\DeclareMathOperator{\norml}{\mathcal{N}}		
\DeclareMathOperator{\ee}{\mathbb{E}}			
\DeclareMathOperator{\prob}{\mathbb{P}}			

\newcommand{\eemph}[1]{\ensuremath{\textbf{#1}}}

\newcommand{\bmat}[1]{\begin{bmatrix}#1\end{bmatrix}}
\newcommand{\probc}[2]{\prob\!\left(#1\,\middle\vert\,#2\right)}	
\newcommand{\eec}[2]{\ee\!\left(#1\,\middle\vert\,#2\right)}		
\newcommand{\eecs}[2]{\ee(#1\,\vert\,#2)}							
\newcommand{\eecc}[2]{\ee\bigl(#1\,\big\vert\,#2\bigr)}							
\newcommand{\T}{\rule{0pt}{2.6ex}}
\newcommand{\hlinet}{\hline\T}

\newcommand{\anc}[1]{{#1^{\uparrow}}}
\newcommand{\des}[1]{{#1^{\downarrow}}}

\newcommand{\funnel}[1]{{#1^{\updownarrow}}}

\newcommand{\sanc}[1]{{#1^{\upuparrows}}}
\newcommand{\sdes}[1]{{#1^{\downdownarrows}}}

\newcommand{\coparent}[1]{{#1^{\vee}}}
\newcommand{\sibling}[1]{{#1^{\wedge}}}
\newcommand{\nonrelative}[1]{{#1^{\sim}}}

\newcommand{\set}[2]{\{#1:#2\}}
\newcommand{\sett}[2]{\bigl\{#1:#2\bigr\}}

\DeclareMathOperator{\lin}{\textbf{lin}}
\newcommand{\defeq}{:=}
\newcommand{\defeqt}{:\eqt}

\usepackage{framed}
\newenvironment{red}{%
  
  \MakeFramed{\advance\hsize-\width \FrameRestore}}
  {\endMakeFramed}

\newcommand{\boxit}[1]{\vspace{5mm}\noindent
\fbox{\begin{minipage}{0.98\linewidth}#1\end{minipage}}\vspace{5mm}

}

\let\VEC      \mathbf
\renewcommand{\vec}{\mathbf}

\begin{document}

\if\MODE3
\title{Optimal Control for LQG Systems on Graphs---Part I: Structural Results}
\else
\title{Optimal Control for LQG Systems on Graphs---Part I: Structural Results}
\fi

\if\MODE3\author{Ashutosh Nayyar \and Laurent Lessard}
\else\author{Ashutosh Nayyar\footnotemark[1] \and
		Laurent Lessard\footnotemark[2]} \fi
\note{Submitted, IEEE Conference on Decision and Control 2013}
\maketitle

\begin{abstract}
In this two-part paper, we identify a broad class of decentralized output-feedback LQG systems for which the optimal control strategies have a simple intuitive estimation structure and can be computed efficiently. Roughly, we consider the class of systems for which the coupling of dynamics among subsystems  and the inter-controller communication is characterized by the same directed graph. Furthermore, this graph is assumed to be a \emph{multitree}, that is, its transitive reduction can have at most one directed path connecting each pair of nodes.
In this first part, we derive sufficient statistics that may be used to aggregate each controller's growing available information. Each controller must estimate the states of  the subsystems that it affects (its \emph{descendants}) as well as the subsystems that it observes (its \emph{ancestors}). The optimal control action for a controller is a linear function of the estimate it computes as well as the estimates computed by all of its ancestors. Moreover, these state estimates may be updated recursively, much like a Kalman filter.
\end{abstract}

\newpage
\tableofcontents
\newpage

\section{Introduction}\label{sec:intro}

With the advent of large scale systems such as the internet or power networks  and increasing applications of networked control systems, the past decade has seen a resurgence of interest in decentralized control.
In decentralized control problems, control decisions must be made using only local or partial information. A key specification of a decentralized control problem is its \emph{information structure}. Intuitively, the information structure of a problem describes for each control decision what information is available for making that decision. The investigation of decentralized control problems available in the current literature has largely been a study of information structures and their implications for the characterization and computation of optimal decentralized control strategies. Two questions that have played a central role in this study are:
\begin{enumerate}
\item[Q1] Can the ever-growing information history available to controllers be aggregated without compromising achievable performance? In other words, are there sufficient statistics for the controllers?
\item[Q2] Can optimal decentralized strategies be efficiently computed?
\end{enumerate}
In this two-part paper, we identify a broad class of information structures and associated decentralized LQG control problems for which both questions have an affirmative answer. We represent information structures  by directed acyclic graphs (DAG), where each node represents both a subsystem and its associated controller, and the edges indicate both the influence of state dynamics between subsystems as well as information-sharing among controllers. An example is given in Figure~\ref{fig:intro_graph}.
\begin{figure}[ht]
\centering
\begin{tikzpicture}[thick,>=latex]
	\tikzstyle{block}=[circle,draw,fill=black!6,minimum height=2.1em]
	\def\x{2.2};
	\def\y{0.8};
	\node [block](N1) at (0,0) {$1$};
	\node [block](N2) at (\x,0) {$2$};
	\node [block](N3) at (.5*\x,-\y) {$3$};
	\node [block](N4) at (1.5*\x,-\y) {$4$};
	\node [block](N6) at (\x,-2*\y) {$5$};
	\draw [->] (N1) -- (N3);
	\draw [->] (N2) -- (N3);
	\draw [->] (N2) -- (N4);
	\draw [->] (N3) -- (N6);
\end{tikzpicture}
\caption{Directed acyclic graph (DAG) representing the information structure of a decentralized control problem. An edge $i\to j$ means that subsystem $i$ affects subsystem $j$ through its dynamics and controller $i$ shares its information with controller $j$.
\label{fig:intro_graph}}
\end{figure}
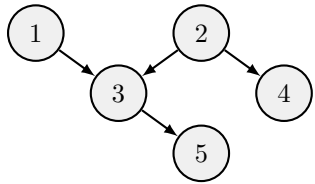

The present paper addresses Q1 by describing sufficient statistics that are both intuitive and familiar; each controller must estimate the states of subsystems that it observes (its ancestors) as well as the states of subsystems that it may potentially affect through its decisions (its descendants). For the example of Figure~\ref{fig:intro_graph}, the controller at node~1 must estimate the states at nodes~$\{1,3,5\}$, while the controller at node~4 must estimate the states at nodes~$\{2,4\}$.  The optimal control action for a controller is a linear function of the estimate it computes as well as the estimates computed by all of its ancestors. In addition to proving the sufficient statistics, we show that they admit a recursive representation, much like a Kalman filter.  In the second paper, we answer Q2 by giving an explicit and efficient method for computing the optimal control strategies.

Our results hold for LQG systems; all dynamics are linear (possibly time-varying), process noise and measurement noise are Gaussian, and the cost function is quadratic over a finite time horizon.  The quadratic cost function may couple certain pairs of nodes and the Gaussian disturbances entering certain pairs of nodes may be correlated. These conditions  are stated precisely in Section~\ref{sec:assumptions}. The associated DAG may be any \emph{multitree}. That is, in the transitive reduction of the DAG, each pair of nodes can be connected by at most one directed path. A key aspect of this work is that we consider \emph{output feedback}. While the presence of measurement noise makes the problem considerably more difficult to solve than the state-feedback case, we will nevertheless see that the optimal controller has a simple and intuitive structure.

\subsection{Prior work}
The available literature on decentralized control problems can be classified into two categories based on the underlying conceptual perspective and the associated techniques employed:
\begin{enumerate}
\item \emph{Decentralized control as a team theory problem:} This perspective of decentralized control problems starts with a description of \emph{primitive random variables} that represent all the uncertainties in the evolution and observations of a dynamic system. The information structure is specified in terms of the \emph{observables} available to each controller, where an observable is a function of the primitive random variables and past decisions. The observables may be explicitly defined in terms of primitive random variables and control decisions or they may be defined using intermediate variables as in a state-space description. Similarly, the control objective can also be viewed as a (either explicit or implicit) function of the primitive random variables and the control decisions. 

The controllers select (measurable) functions that map their information to their decisions.
Once such a selection of strategies has been made, the control objective becomes a well-defined random variable. The design problem is to identify a choice of strategies that minimizes the expected value of the control objective. 
\item \emph{Decentralized control as closed-loop norm optimization:} This perspective focuses on cases when the plant and the controllers are linear time-invariant (LTI) systems. The plant has two inputs: an exogenous input vector~$w$ and a control input vector~$u$. The plant has two outputs: a performance-related vector~$z$ and an observation vector~$y$. The controller is a LTI system with $y$ as its input and $u$ as its output. The information structure is described in terms of structural constraints on the transfer function of the controller. For a fixed choice of the controller, the closed loop LTI system  can be described in terms of its {transfer function} from the exogenous input $w$ to the performance related vector $z$. The design problem is to minimize a norm (such as the $\mathcal{H}_2$-norm) of this transfer function.
\end{enumerate}

Both of these approaches have acknowledged the difficulty of a general decentralized control problem with arbitrary information structure \cite{blondel,witsenhausen}. Therefore, there has been considerable interest in identifying classes of information structures that may be ``easier'' to solve. In the team-theoretic approach, partial nestedness of an information structure has been identified as a key simplifying feature~\cite{hochu}. A decentralized LQG control problem with a partially nested (PN) information structure admits a linear control strategy as the globally optimal strategy choice. Further, it can be reduced (at least for finite horizon problems) to a static LQG team problem for which person-by-person optimal strategies are globally optimal~\cite{radner}. While the results of~\cite{hochu} have been generalized to certain infinite-horizon problems~\cite{mahajan_infinite}, a universal and computationally efficient methodology for finding optimal strategies for all partially nested problems remains elusive.

In the norm optimization framework, some properties of the plant and the information constraint have been identified as simplifying conditions. These properties imply convexity of the transfer function norm optimization problem in decentralized control. Examples of such properties include quadratic invariance (QI)~\cite{rotkowitz06}, funnel causality~\cite{bamieh_funnel}, and certain hierarchical architectures~\cite{qimurti04}. Interestingly, many of the systems considered in the literature using these properties have a partially nested information structure as well. Despite the convexity, the optimization problem in general is infinite dimensional and therefore hard to solve.

For decentralized problems that do not belong to the class of PN or QI problems, linearity of  globally optimal strategies is not guaranteed and the optimization of strategies is in general a non-convex problem. These problems present unique features  such as signaling (that is, communication through control actions) that make them particularly difficult. We refer the reader to \cite{grover,mahajan_survey,nayyar, witsenhausen, yuksel} for some examples and results for such problems.

The decentralized control problems considered in this paper are a subset of the class of PN and QI problems.   The relation between problem classes is illustrated in Figure~\ref{fig:venn}.
\begin{figure}[ht]
\centering
\begin{tikzpicture}[thick,>=latex]
	\def\y{3mm}
	\def\x{1.5mm}
	\tikzstyle{rbox}=[rounded corners,draw,fill=black!6,minimum height=2.1em,anchor=east]
	\node [rbox,text width=126mm,minimum height=15mm+3*\y](N1) at (3*\x,0) {%
		\parbox{25mm}{\textbf{Decentralized:}\\
		\small
		\textbullet\,{\color{red}not linear}\\
		\textbullet\,{\color{red}not convex}\\
		\textbullet\,{\color{red}generally hard}}};
	\node [rbox,text width=95mm,minimum height=15mm+2*\y](N2) at (2*\x,0) {%
		\parbox{25mm}{\textbf{PN and QI:}\\
		\small
		\textbullet\,linear optimal\\
		\textbullet\,convexifiable\\
		\textbullet\,\color{red}{$\infty$-dimensional}}};
	\node [rbox,text width=65mm,minimum height=15mm+\y,fill=white](N3) at (\x,0) {%
		\parbox{30mm}{\textbf{Multitree:}\\
		\small
		\textbullet\,sufficient statistics\\
		\textbullet\,structural result\\
		\textbullet\,efficient solution}};
	\node [rbox,minimum height=15mm](N4) at (0,0) {%
		\parbox{31mm}{\textbf{Centralized:}\\
		\small
		\textbullet\,certainty equivalent \\
		\textbullet\,separation principle}};
\end{tikzpicture}
\caption{Venn diagram showing a complexity hierarchy of decentralized control problems. The present paper establishes the \emph{multitree} class, which shares many structural characteristics with centralized problems, but is more computationally tractable than general partially nested and quadratically invariant problems.\label{fig:venn}}
\end{figure}
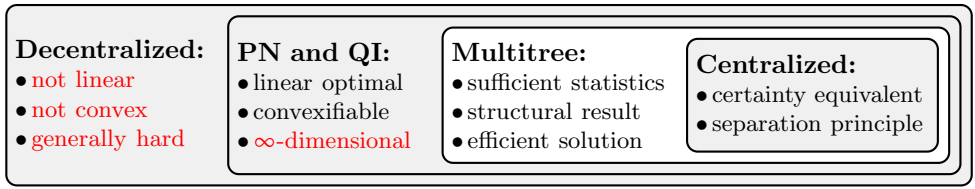

In decentralized control problems with multiple subsystems that each have an associated controller, partial nestedness is typically manifested in two ways: 
\begin{enumerate}
\item If a subsystem $i$ affects another subsystem $j$, then the controller at subsystem $i$ shares all information with the controller at subsystem $j$. In other words, the information flow obeys the \emph{sparsity constraints} of the dynamics. State-feedback problems of this kind were considered in~\cite{swigart10} for a two-controller case and in~\cite{shah_parrilo} for controllers communicating over a DAG. Similar results and partial extensions to output feedback were addressed in \cite{jonghan_separable,swigart_partial,jain}. The first solution to a problem with output feedback at every subsystem appeared in~\cite{lessard_allerton,lessard_acc} for a DAG with two subsystems. The two-subsystem output feedback results were generalized to star-shaped (broadcast) graphs in~\cite{lessard_broadcast} and linear chain graphs in~\cite{tanaka_triangular}.

\item If subsystem $i$ affects subsystem $j$ after some delay $d$, then the controller at subsystem $i$ shares its information with controller at subsystem $j$ with delay not exceeding $d$. In other words, the information flow obeys the \emph{delay constraints} of the dynamics. Among the earlier problems of this type were the discrete-time LQG problems with  one-step delayed sharing information structure~\cite{kurtaran:1974, nonclassical,WalrandVaraiya:1978,yoshikawa}. More recently, problems where communication delay between controllers is characterized in terms of  distance on a DAG were addressed in \cite{lamperski_delayed_recent} for the state feedback case and in  \cite{lamperski_outputfeedback} for the output feedback case.  In both these works, it is assumed that the underlying DAG is \emph{strongly connected}, so every measurement eventually finds its way to every decision-maker. This fact allows the optimal strategy to be decomposed into a component acting on the common past information together with a finite impulse response (FIR) portion acting on newer information. A similar model with state feedback but completely general DAG was addressed in \cite{lamperski_lessard}. The state-feedback assumption allows the control problem to decouple into smaller centralized problems that can be solved separately.

\end{enumerate}

The output-feedback works with sparsity constraints~\cite{lessard_broadcast,lessard_allerton,lessard_acc,tanaka_triangular} use a closed-loop norm optimization framework and employ a spectral factorization approach to solve for the optimal controller. This approach yields an observer-controller structure and shows how to jointly solve for the appropriate estimation and control gains. While this approach does not provide an immediate way to interpret the states of the optimal controller as minimum mean-squared error (MMSE) estimates of the plant state, such an interpretation was obtained for the two-subsystem case in~\cite{lessard_tpof_TAC}.

An alternative approach is to use the team-theoretic perspective. The two-subsystem output feedback problem was solved in this manner in~\cite{lessard_nayyar_tpof}. We use a similar approach in the present work to extend the output feedback results to a broader class of DAGs. The advantage of a team-theoretic approach is that structural results emerge naturally, and one can deduce the optimal controller's sufficient statistics without solving for gains explicitly. Indeed, the paper~\cite{lessard_nayyar_tpof} derives structural results for a finite-horizon formulation with a linear time-varying plant, whereas the works~\cite{lessard_broadcast,lessard_allerton,lessard_acc,lessard_tpof_TAC,tanaka_triangular} address linear time-invariant plants and find the infinite-horizon steady-state optimal controller.

\subsection{Organization}

In Section~\ref{sec:prelim}, we explain notations, conventions, and assumptions. Section~\ref{sec:main} presents the main structural result as well as some examples. The proof spans Sections~\ref{sec:proof_prelim}, \ref{sec:proof_main} and \ref{sec:proof_alt}. We conclude in Section~\ref{sec:conclusion}.





\section{Preliminaries}\label{sec:prelim}

\subsection{Basic notation}

Real vectors and matrices are represented by lower- and upper-case letters respectively. Boldface symbols denote random vectors, and their non-boldface counterparts denote particular realizations. $x^\tp$ denotes the transpose of vector $x$. The probability density function of $\vec{x}$ evaluated at $x$ is denoted $\prob(\vec{x}=x)$, and conditional densities are written as $\probc{\vec{x}}{\vec{y}=y}$. $\mathbb{E}$ denotes the expectation operator. We write
$
\vec{x} \sim \norml(\mu,\Sigma)
$
when $\vec{x}$ is normally distributed with mean $\mu$ and covariance $\Sigma$. 

We consider discrete time stochastic processes over a finite time interval $[0,T]$.
Time is indicated using subscripts, and we use the colon notation to denote ranges. For example:
\[
x_{0:T-1} = \{ x_0, x_1,\dots, x_{T-1} \}
\]
In general, all symbols are time-varying. In an effort to present general results while keeping equations clear and concise, we introduce a new notation to represent a family of equations. For example, when we write:
\[
\vec{x}_+ \eqt A\vec{x} + \vec{w},
\]
we mean that $\vec{x}_{t+1} = A_t \vec{x}_t + \vec{w}_t$ holds for $0 \le t \le T-1$. Note that the subscript ``$+$'' indicates that we increment to $t+1$ for the associated symbol. We similarly overload the summation symbol by writing for example
\[
\sum_t x^\tp Q x 
\quad\text{to mean}\quad \sum_{t=0}^{T-1} x_t^\tp Q_t x_t
\]
Whenever the symbol $t$ is written above a binary relation or below a summation, it is implied that $0 \le t \le T-1$. There is no ambiguity because we use the same time horizon~$T$ throughout this paper.


We denote subvectors by using superscripts. 
 Subvectors may also be referenced by  using a subset of indices as superscripts. For example, for a vector 
\[
\vec{x} =   \bmat{\vec{x}^1 \\ \vec{x}^2 \\ \vec{x}^3}
\] 
and $s = \{1,3\}$, we will use the concise notation
\[
\vec{x}^s = \vec{x}^{\{1,3\}} = \bmat{\vec{x}^1 \\ \vec{x}^3}
\]
When writing sub-vectors, we will always arrange the components in increasing order of indices. Thus, in the above example, $\vec{x}^s$ is $\bmat{\vec{x}^1 \\ \vec{x}^3}$ and not $\bmat{\vec{x}^3 \\ \vec{x}^1}$. Given a collection of random vectors, we will at times treat the collection as a concatentation of vectors arranged in increasing order of node index.

For a matrix $A$, $A_{ij}$ denotes its $(i,j)$-block whose dimensions are inferred from the context.  Given two sets of indices $s$ and $r$, $A^{s,r}$ is a matrix composed of blocks $A_{ij}$ with $i \in s$ and $j \in r$. The blocks $A_{i,j}$ in a row (column) of $A^{s,r}$ are arranged in increasing order of column (row) indices.
 

We will write $a \in \lin(p_1,\dots,p_m)$ to mean that~$a$ is a linear function of $p_1,\dots,p_m$. In other words, $a = A_1 p_1 + \dots + A_m p_m$ for some appropriately chosen matrices $A_1,\dots,A_m$.

\subsection{Graphs}

Let $G(\mathcal{V},\mathcal{E})$ be a directed acyclic graph (DAG). The nodes are labeled  $1$ to $n$, so $\mathcal{V}=\{1,\dots,n\}$. If there is an edge from $i$ to $j$, we write $(i,j)\in \mathcal{E}$.

We write $i\to j$ if there is a \emph{directed path} from $i$ to $j$. That is, if there exists a sequence of nodes $v_1,\dots,v_m$ with $v_1=i$ and $v_m=j$ such that $(v_k,v_{k+1})\in\mathcal{E}$ for all~$k$. By convention, every node has a directed path (of length zero) to itself. So it is always true that $i\to i$. We write $i \leftrightarrow j$ if $i$ and $j$ are \emph{path-connected}, that is, if $i\to j$ or $j \to i$. Otherwise, we say they are \emph{path-disconnected}, and we write $i \nleftrightarrow j$. We can express the path-connectedness of $G$ using the \emph{sparsity matrix}, which is the binary matrix $\Sbin\in\{0,1\}^{n\times n}$ defined by
\[
\Sbin_{ij} = \begin{cases}
1 &\text{if } j \to i \\
0 &\text{otherwise}
\end{cases}
\]
Note that different graphs may have the same sparsity matrix. In general, $\Sbin$ is the adjacency matrix of the transitive closure of $G$. So graphs with the same transitive closure also share the same sparsity matrix.

By convention, we assign a \emph{topological ordering} to the node labels.
 That is, we choose a labeling such that if $j \to i$, then $j \le i$. This is possible for any DAG  \cite[\S 22.4]{cormen2001}. Therefore, $\Sbin$ is always lower-triangular. See Figure~\ref{fig:ex_simplegraph} for a simple example of a DAG and its associated sparsity matrix.

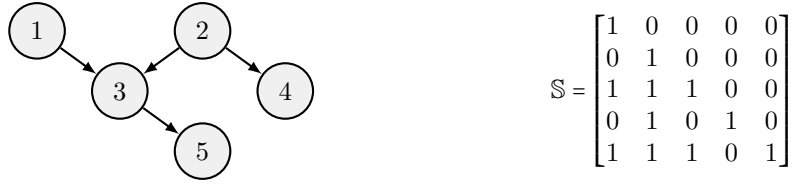
\begin{figure}[ht]
\centering
\begin{subfigure}{0.49\linewidth}
\centering
\begin{tikzpicture}[thick,>=latex]
	\tikzstyle{block}=[circle,draw,fill=black!6,minimum height=2.1em]
	\def\x{2.2};
	\def\y{0.8};
	\node [block](N1) at (0,0) {$1$};
	\node [block](N2) at (\x,0) {$2$};
	\node [block](N3) at (.5*\x,-\y) {$3$};
	\node [block](N4) at (1.5*\x,-\y) {$4$};
	\node [block](N6) at (\x,-2*\y) {$5$};
	\draw [->] (N1) -- (N3);
	\draw [->] (N2) -- (N3);
	\draw [->] (N2) -- (N4);
	\draw [->] (N3) -- (N6);
\end{tikzpicture}
\end{subfigure}
\begin{subfigure}{0.49\linewidth}
\centering
$\Sbin = \bmat{ 1 & 0 & 0 & 0 & 0  \\
			0 & 1 & 0 & 0 & 0  \\
			1 & 1 & 1 & 0 & 0  \\
			0 & 1 & 0 & 1 & 0  \\
			1 & 1 & 1 & 0 & 1 }$
\end{subfigure}
\caption{Simple DAG and its associated sparsity matrix.
\label{fig:ex_simplegraph}}
\end{figure}
Given a node $i\in \mathcal{V}$, we define its \emph{ancestors} as the set of nodes that have a directed path to $i$. Similarly, we define the \emph{descendants} as the set of nodes that $i$ can reach via a directed path. We use the following notation for ancestors and descendants respectively.
\begin{align*}
\anc{i} = \set{j\in \mathcal{V}}{j \to i}
&&
\des{i} = \set{j\in \mathcal{V}}{i \to j}
\end{align*}
Ancestors and descendants of $i$ always include $i$ itself.
We define the \emph{strict ancestors} and \emph{strict descendants} when we mean to exclude $i$. Specifically, $\sanc{i} = \anc{i}\setminus \{i\}$ and $\sdes{i} = \des{i}\setminus\{i\}$.
We use the notation $\funnel{i} = \anc{i} \cup \des{i}$ for the set of all nodes that are path-connected to node $i$. Note that $\funnel{i}$ is partitioned as $\sanc{i} \cup \{i\} \cup \sdes{i}$. In the graph of Figure~\ref{fig:ex_simplegraph}, for example,
\begin{align*}
\anc{3} = \{1,2,3\} &&
\sanc{3} = \{1,2\} &&
\des{3} = \{3,5\} &&
\sdes{3} = \{5\} &&
\funnel{3} = \{1,2,3,5\}
\end{align*}

We summarize these notations in Table \ref{table:notation1}.

\begin{table}
\centering
\begin{tabular}{|c|c|}
\hlinet
Notation & Meaning \\
\hlinet
$\anc{i}$ & $\set{j\in \mathcal{V}}{j \to i}$\\
\hlinet
$\des{i}$ & $\set{j\in \mathcal{V}}{i \to j}$\\
\hlinet
$\funnel{i}$ & $\anc{i} \cup \des{i}$\\
\hlinet
$\sanc{i}$ & $\anc{i}\setminus i$\\
\hlinet
$\sdes{i}$ & $\des{i} \setminus i$\\
\hline
\end{tabular}
\caption{Ancestors, descendants and related definitions}
\label{table:notation1}
\end{table}

\begin{rem}
Note that while $\anc{i}, \des{i}$, etc. are defined as subsets, it is convenient to think of them as ordered lists in which the node indices are arranged in increasing order. Thus, in Figure~\ref{fig:ex_simplegraph}, $\anc{3}$ will always be written as $\{1,2,3\}$ and not as any other permutation of $\{1,2,3\}$.
\end{rem}
\begin{rem}
A node with no strict descendants is called a {leaf node} and a node with no strict ancestors is called a root node.
\end{rem}


\subsection{System model}

The system we consider consists of $n$ subsystems that may affect one another according to the structure of an underlying DAG, $G(\mathcal{V},\mathcal{E})$. The $i^\text{th}$ subsystem has state~$\vec{x}^i$, input~$\vec{u}^i$, measurement~$\vec{y}^i$, process noise~$\vec{w}^i$, and measurement noise~$\vec{v}^i$. We assume these are discrete-time random processes that satisfy the following state-space dynamics for all $i\in \mathcal{V}$.
\begin{equation}\label{eq:state_eqns_sum}
\begin{aligned}
\vec{x}^i_+ &\eqt \sum_{j \in \anc{i}} \bigl(A_{ij}\vec{x}^j + B_{ij}\vec{u}^j\bigr) + \vec{w}^i\\
\vec{y}^i &\eqt \sum_{j\in \anc{i}} \bigl(C_{ij} \vec{x}^j\bigr) + \vec{v}^i
\end{aligned}
\end{equation}
The relative timing of $i^\text{th}$ state, control action and observation at time $t$ is as shown in Figure \ref{fig:timing}. Note that the observation $\vec{y}^i_t$ is generated after control action $\vec{u}^i_t$ is taken.
\begin{figure}[ht]
\centering
\begin{tikzpicture}[thick,>=latex,scale=1.2]
\def\dx{0.3}  
\def\dxx{0.5} 
\def\dy{0.16}  
\def\dyy{0.22} 
\draw[dashed] (-\dxx-\dx,0) -- (-\dx,0);
\draw[->] (-\dx,0) -- (4+2*\dx,0);
\draw (0,\dyy) -- (0,-\dyy) node[anchor=north]{$t$};
\draw (4,\dyy) -- (4,-\dyy) node[anchor=north]{$t+1$};
\draw (1,-\dy) -- (1,\dy) node[anchor=south]{$\vec x_t^i$};
\draw (2,-\dy) -- (2,\dy) node[anchor=south]{$\vec u_t^i$};
\draw (3,-\dy) -- (3,\dy) node[anchor=south]{$\vec y_t^i$};
\end{tikzpicture}
\caption{Relative timing of state, action and observation at time $t$}
\label{fig:timing}
\end{figure}
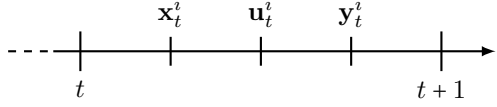
Each of the matrices in \eqref{eq:state_eqns_sum} may be time-varying, and may even change dimensions with time. In an effort to make our notation more concise, we concatenate the various symbols above and simply write
\begin{equation}\label{eq:state_eqns}
\begin{aligned}
\vec{x}_+ &\eqt A\vec{x} + B\vec{u} + \vec{w} \\
\vec{y} &\eqt C\vec{x} + \vec{v}
\end{aligned}
\end{equation}
In this condensed notation, the matrices $A$, $B$, and $C$ have blocks that conform to $\Sbin$, the sparsity matrix for $G(\mathcal{V},\mathcal{E})$. In other words, if $\Sbin_{ij} = 0$ then $A_{ij}=0$, $B_{ij}=0$, and $C_{ij}=0$.
The random vectors in the collection
\begin{equation}\label{eq:indep_noise}
\left\{
	 \vec{x}_0, \bmat{ \vec{w}_0\\ \vec{v}_0}, \dots, \bmat{ \vec{w}_{T-1}\\ \vec{v}_{T-1}}
\right\},
\end{equation}
referred to as the primitive random variables, are mutually independent and jointly Gaussian with the following known probability density functions
\begin{equation}\label{eq:initial_conditions}
\begin{aligned}
\vec{x}_0 & = \norml(0, \Sigma_\textup{init}) \\
\bmat{\vec{w}\\\vec{v}} & \eqt \norml\biggl(0,\bmat{W & U^\tp \\ U & V}\biggr)
\end{aligned}
\end{equation}
There are $n$ controllers, one responsible for each of the $\vec{u}^i$, $i=1,\ldots,n$.  We define the \emph{locally generated information} at node $i$ at time $t$ as
\begin{equation}\label{eq:information_pattern}
	\vec{i}_t^i := \Bigl\{ \vec{y}_{0:t-1}^{{i}}, \vec{u}_{0:t-1}^{{i}} \Bigr\}
\end{equation}
The information available to controller $i$ at time $t$ is 
\begin{equation}\label{eq:information_pattern2}
 \vec{i}^\anc{i}_t := \bigcup_{j \in \anc{i}} \vec{i}_t^j =  \bigcup_{j \in \anc{i}} \Bigl\{ \vec{y}_{0:t-1}^{{j}}, \vec{u}_{0:t-1}^{{j}} \Bigr\} .
\end{equation}
In other words, each controller knows the past measurements and decisions of its ancestors. So the directed edges of $G$ may be thought of as representing the \emph{flow of information} between subsystems. Crucially, the same graph $G$  represents both how the dynamics propagate as well as how information is shared in the system. The controllers select actions according to control strategy ${f}^i :=(f^i_0,f^i_1,\ldots,f^i_{T-1})$ for $i\in\mathcal{V}$. That is,
\begin{equation}\label{eq:input_definition}
	\vec{u}_t^i = f^i_t (\,\vec{i}^\anc{i}_t \,)
\qquad\text{for }0\le t\le T-1
\end{equation}

Given a control strategy profile ${f} = ({f}^1,{f}^2,\dots,{f}^n)$, performance is measured by the finite horizon expected quadratic cost
\begin{equation}\label{eq:cost_np}
\hat{\mathcal{J}}_0(f) =
\ee^{f} \biggl( \sum_t
\bmat{\vec{x}\\\vec{u}}^\tp \bmat{Q & S \\ S^\tp & R} \bmat{\vec{x}\\ \vec{u}}
+ \vec{x}_T^\tp P_\textup{final} \vec{x}_T \biggr)
\end{equation}
The expectation is taken with respect to the joint probability measure on $(\vec{x}_{0:T},  \vec{u}_{0:T-1})$ induced by the choice
of $f$. 

It is assumed that all system parameters are universally known. Specifically, $\Sigma_\text{init},P_\text{final}$, as well as the values of $A,B,C,Q,R,S,U,V,W$ for all $t$, are known by all controllers.

\subsection{Assumptions}\label{sec:assumptions}

In addition to the problem specifications~\eqref{eq:state_eqns}--\eqref{eq:cost_np}, we will make some additional assumptions about the underlying DAG and the noise and cost parameters used in~\eqref{eq:initial_conditions} and~\eqref{eq:cost_np} respectively. First, we require some definitions.

\begin{defn}[multitree]\label{def:multitree}
The nodes $i,a,b,j\in\mathcal{V}$ form a \eemph{diamond} if $i\to a\to j$ and $i\to b\to j$, and $a\nleftrightarrow b$. A \eemph{multitree} is a directed acyclic graph that contains no diamonds.
\end{defn}
\noindent For example, the graph of Figure~\ref{fig:ex_simplegraph} is a multitree. However, if we add the edge $(4,5)$, then the nodes $(2,3,4,5)$ form a diamond and the graph ceases to be a multitree.

\begin{defn}[decoupled cost]\label{def:indep_cost}
Let $\mathcal{X} = \{Q_{0:T-1},R_{0:T-1},S_{0:T-1},P_\textup{final}\}$ denote the set of matrices associated with the cost function. We say that nodes $i,j\in\mathcal{V}$ have \eemph{decoupled cost} if $X_{ij}=0$ for all $X\in\mathcal{X}$.
\end{defn}

\begin{defn}[uncorrelated noise]\label{def:indep_noise}
Let $\mathcal{Y} = \{W_{0:T-1},V_{0:T-1},U_{0:T-1},\Sigma_\textup{init}\}$ denote the set of matrices associated with the noise and initial state statistics. We say that nodes $i,j\in\mathcal{V}$ have \eemph{uncorrelated noise} if $Y_{ij}=0$ for all $Y\in\mathcal{Y}$.
\end{defn}

The notions of decoupled cost and uncorrelated noise have intuitive interpretations. If two nodes have decoupled cost, then the instantaneous cost at any time has no ``cross-terms'' that involve both the nodes. If two nodes have uncorrelated noise, then the process and measurement noises affecting one node are statistically independent of those affecting the other. Our assumptions are as follows.

\begin{list}{}{}

\item[\textbf{(A1)}] The DAG $G(\mathcal{V,E})$ is a multitree.

\item[\textbf{(A2)}] For every pair of nodes $i,j\in\mathcal{V}$, we have some combination of decoupled cost and uncorrelated noise, depending on whether these nodes have a common ancestor or a common descendant. We state the requirements in the following table.
\begin{center}
\begin{tabular}{|l|l|p{3.5cm}|}
\hlinet
$i$ and $j$ have... & a common ancestor & no common ancestor \\
\hlinet
a common descendant & no restrictions & uncorrelated noise \\
\hlinet
no common descendant & decoupled cost & decoupled cost \textbf{and}\newline uncorrelated noise  \\
\hline
\end{tabular}
\end{center}

\item[\textbf{(A2')}] This assumption is a relaxed version of A2, defined by the following table.
\begin{center}
\begin{tabular}{|l|l|p{3.5cm}|}
\hlinet
$i$ and $j$ have... & a common ancestor & no common ancestor \\
\hlinet
a common descendant & no restrictions & uncorrelated noise \\
\hlinet
no common descendant & decoupled cost & decoupled cost \textbf{or}\newline uncorrelated noise  \\
\hline
\end{tabular}
\end{center}
 
\end{list}
For the graph of Figure~\ref{fig:ex_simplegraph}, nodes~1 and~4 have neither a common ancestor nor a common descendant. Therefore, A2 would require that this pair of nodes have both decoupled cost and uncorrelated noise. However,~A2' would only require one of the two.

Note that because of the multitree assumption, the only way~$i$ and~$j$ can have both a common ancestor \emph{and} a common descendant is if $i\leftrightarrow j$. Assumption A2 may be expressed in terms of the sparsity pattern $\Sbin$ using the following observations
\begin{enumerate}
\item $(\Sbin\Sbin^{\tp})_{ij} = 0$ if and only if $\anc{i}\cap\anc{j}=\emptyset$ ($i$ and $j$ have no common ancestor)
\item $(\Sbin^\tp \Sbin)_{ij} = 0$ if and only if $\des{i}\cap\des{j}=\emptyset$ ($i$ and $j$ have no common descendant)
\end{enumerate}
So A2 may be stated concisely as follows: all matrices in $\mathcal{X}$ (see Definition \ref{def:indep_cost}) have the same sparsity as $\Sbin^\tp \Sbin$ and all matrices in $\mathcal{Y}$ (see Definition \ref{def:indep_noise}) have the same sparsity as $\Sbin\Sbin^\tp$. For the graph of Figure~\ref{fig:ex_simplegraph}, these sparsity patterns are
\begin{align*}
\Sbin\Sbin^\tp \sim \bmat{ 1 & 0 & 1 & 0 & 1 \\
			0 & 1 & 1 & 1 & 1 \\
			1 & 1 & 1 & 1 & 1 \\
			0 & 1 & 1 & 1 & 1 \\
			1 & 1 & 1 & 1 & 1 }
&&
\Sbin^\tp \Sbin \sim \bmat{ 1 & 1 & 1 & 0 & 1 \\
			1 & 1 & 1 & 1 & 1 \\
			1 & 1 & 1 & 0 & 1 \\
			0 & 1 & 0 & 1 & 0 \\
			1 & 1 & 1 & 0 & 1 }
\end{align*}
Assumption A2' is less restrictive than A2, but cannot be as easily expressed in terms of the sparsity pattern $\Sbin$. 


\begin{rem}
Note that both assumptions A2 and A2' are more general than the assumption  that all cost matrices in $\mathcal{X}$ and covariance matrices in $\mathcal{Y}$ are block-diagonal.
\end{rem}


\section{Main result and examples}\label{sec:main}

The problem addressed in this paper is as follows.

\boxit{
\begin{problem}[$n$-player LQG]
  \label{prob:NPLQG}
  For the model~\eqref{eq:state_eqns}--\eqref{eq:input_definition}, and subject to Assumptions A1 and either A2 or A2', find a control strategy profile $f = (f^1,f^2,\ldots,f^n)$
  that minimizes the expected cost~\eqref{eq:cost_np}.
\end{problem}
}

The information structure of our problem is \emph{partially nested} and therefore, without loss of optimality, we will restrict attention to linear control strategies \cite{hochu}.

The main result of this paper is a description of sufficient statistics required for an optimal solution of Problem~\ref{prob:NPLQG}. We first define the following conditional expectation:

\[\mathbf{z}^{\funnel{j}}_t \defeq \eecs{\vec{x}_t^{\funnel{j}}}{\vec{i}^\anc{j}_t}. \]
In other words, $\mathbf{z}^{\funnel{j}}_t$ is the conditional mean of the state of all nodes that are path connected to node $j$ based on the information available to node $j$ (recall the notations defined in Table~\ref{table:notation1}). Our first result is the following.

\begin{thm}[Control Result]\label{thm:main}
In Problem~\ref{prob:NPLQG}, there is no loss in optimality in jointly restricting all nodes $i \in \mathcal{V}$ to strategies of the form
\begin{align}
  \VEC u^i_t &= \sum_{j \in \anc{i}}K^{ij}_t\mathbf{z}^{\funnel{j}}_t  \label{eq:mainresult}
\end{align}
where
$\mathbf{z}^{\funnel{j}}_t \defeq \eecs{\vec{x}_t^{\funnel{j}}}{\vec{i}^\anc{j}_t}$.
\end{thm}

Recall that $\vec{i}_t^\anc{i}$ defined in \eqref{eq:information_pattern}--\eqref{eq:information_pattern2} is the information available to node~$i$ at time~$t$, and this set grows with time as more measurements are observed and more decisions are made. Theorem~\ref{thm:main} states that controllers need not remember this entire information history. Instead, each node $j$ may compute the aggregated statistic $\vec z^\funnel{j}$, which is an estimate of the current states of its ancestors and descendants. The optimal decision at node $i$, $\vec u^i$, is then a linear function of the estimates maintained by all of its ancestors.



An alternative way of stating the result of Theorem~\ref{thm:main} is to stack the decisions and estimates and obtain one large linear equation.

\begin{cor}\label{cor:main}
In Problem~\ref{prob:NPLQG}, there is no loss in optimality in jointly restricting the strategies of all the nodes to the form
\begin{equation}\label{eq:matrix_form}
\bmat{\vec{u}_t^1 \\ \vdots \\ \vec{u}_t^n} =
\underbrace{\bmat{ K_t^{11} & \cdots & K_t^{1n} \\
	   \vdots & \ddots & \vdots \\
	   K_t^{n1} & \cdots & K_t^{nn} }}_{K_t}
\bmat{\eecs{\vec{x}_t^\funnel{1}}{\vec{i}_t^\anc{1}} \\
\vdots \\ \eecs{\vec{x}_t^\funnel{n}}{\vec{i}_t^\anc{n}} }
\qquad\text{for}\quad 0\le t \le T-1
\end{equation}
where $K_t$ is a  matrix with block-sparsity conforming to $\Sbin$.
\end{cor}

Our second result addresses the evolution of the estimates $\vec{z}^{\funnel{j}}_t$. In order to state this result, we need to define the following linear operation. 
\begin{defn}\label{def:Gmatrix}
Consider nodes $i$ and $j$ with $i \in \sdes{j}$.  Let $\funnel{i} = \{k_1,k_2,\ldots,k_{|\funnel{i}|}\}$ and $\funnel{j} = \{l_1,l_2,\ldots,l_{|\funnel{j}|}\}$.  Define a matrix $E^{i,j}$ with $|\funnel{i}|$ block rows and $|\funnel{j}|$ block columns as follows: For $a = 1,2,\ldots, |\funnel{i}|$,
\begin{enumerate}
\item If $k_a \notin \funnel{j}$, then the $a^\text{th}$ block row of $E^{i,j}$ is $0$.
\item If $k_a \in \funnel{j}$ and $k_a = l_b$, then the $(a,b)$ block of $E^{i,j}$ is identity and the rest of $a^\text{th}$ block row is $0$.
\end{enumerate}
\end{defn}
For example,  in Figure~\ref{fig:ex_simplegraph},  $\funnel{3}=\{1,2,3,5\}$ and $\funnel{2}=\{2,3,4,5\}$, then 
\[E^{3,2} = \bmat{  0 & 0 & 0 & 0  \\
			I & 0 & 0 & 0  \\
		    0 & I & 0 & 0  \\
			 0 & 0 & 0 & I   }. \] 
We can now state our second result.
\begin{thm}[Estimation Result]\label{thm:main_2}
If the control strategy is as given in Theorem \ref{thm:main}, then the evolution of $\vec{z}^{\funnel{j}}_t$ is described as follows:
\begin{equation}
\begin{aligned}
\vec{z}^{\funnel{j}}_0 &= 0 \\
\vec{z}^{\funnel{j}}_+ &\eqt
A^{\funnel{j}\funnel{j}}\vec{z}^{\funnel{j}}
 + {B}^{\funnel{j}\funnel{j}}\bmat{ \vec{u}^{\anc{j}} \\ \{\vec{\hat u}^{ij}\}_{i \in \sdes{j}}} - {L}^j\left(\vec{y}^{\anc{j}} -C^{\anc{j}\anc{j}}\vec{z}^{\anc{j}}\right)
\end{aligned}
\end{equation}
for some matrices ${L}^j_t$, $0 \leq t \leq T-1$, with
\begin{equation}\label{eq:main_2eq1}
\vec{\hat u}^{ij} \eqt \sum_{a \in \anc{j}} K^{ia}\vec{z}^\funnel{a} + \sum_{b \in \anc{i} \cap \sdes{j}} K^{ib}E^{b,j}\vec{z}^\funnel{j} \qquad\text{for}\quad i \in \sdes{j}.
\end{equation}
\end{thm}

\begin{rem}
For linear control strategies described by Theorem \ref{thm:main_2}, $\vec{\hat u}^{ij}_t$ as defined in \eqref{eq:main_2eq1} is in fact equal to $\eecc{\vec{u}^i_t}{\vec{i}^{\anc{j}}_t}$. 
\end{rem}

The above theorems provide  finite dimensional sufficient statistics for all controllers in the system. These results should be viewed as structural results of optimal control strategies since they postulate the existence of optimal controllers and estimators of the form presented above without specifying how  the matrices $K^{ij}, L^j$ used in control and estimation can be computed. 

A detailed proof of Theorems~\ref{thm:main} and \ref{thm:main_2}  will be developed in Sections~\ref{sec:aggregated}--\ref{sec:proof_main}. In the remainder of this section, we will give examples that illustrate the applicability of our main result.

\newcounter{EXcount}
\begin{list}{\textbf{Example~\arabic{EXcount}.}}{\usecounter{EXcount}}

\item (Centralized case) The classical LQG problem is the case where $\mathcal{V} =\{1\}$. Theorem~\ref{thm:main} then yields the classical separation result; $\vec{u}_t = K_t \eec{\vec{x}_t}{\vec{i}_t}$ and Theorem~\ref{thm:main_2} reduces to the standard Kalman filter.

\item (Disconnected systems) Consider several subsystems with decoupled dynamics. In  this case, the underlying graph has no edges and $\Sbin = I$. Under Assumption~A2', if each pair of nodes has either decoupled cost or uncorrelated noise, then by Theorem~\ref{thm:main}, the optimal strategy for  node $i$ is of the form $\vec{u}_t^i = K_t^i \eecs{\vec{x}^i_t}{\vec{i}^i_t}$. So node $i$ need only estimate its local state and by Theorem~\ref{thm:main_2}, this estimate is updated by a standard Kalman filter. This result can be intuitively explained as follows:
\begin{enumerate}
\item If nodes $i$ and $j$ have uncorrelated noise but coupled cost, there is value in estimating $\vec{x}^j$, but this estimate is zero because of the uncoupled dynamics and uncorrelated noise.
\item If nodes $i$ and $j$ have decoupled cost but correlated noise, the estimate of $\vec{x}^j$ would be non-zero, but the information it provides is of no use because the state and decisions $\vec{x}^j$ and $\vec{u}^j$ do not affect the cost associated with node~$i$.
\end{enumerate}
\begin{rem}
If we violate Assumption~A2' by allowing both coupled cost and correlated noise, then Theorem~\ref{thm:main} no longer holds. An instance of such a problem was solved for the state feedback case~\cite{lessard_diagsf,lessard_diagsf2}, and it was found that depending on the system parameters, the optimal infinite-horizon controller may have up to $n^2$ states, where $n$ is the global plant's dimension. This result suggests a new controller that lies outside the scope of Theorem~\ref{thm:main}.
\end{rem}
\item (Two-player and linear chain cases) Consider the two-player problem corresponding to the graph and sparsity pattern given below.
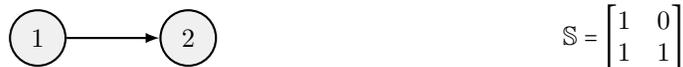
\begin{figure}[ht]
\centering
\begin{subfigure}{0.49\linewidth}
\centering
\begin{tikzpicture}[thick,node distance=2.0cm,>=latex]
	\def\fs{\small}  
	\tikzstyle{block}=[circle,draw,fill=black!6,minimum height=2.1em]
	\node [block](N1){$1$};
	\node [block, right of=N1](N2){$2$};
	\draw [->] (N1) -- (N2);;
\end{tikzpicture}
\end{subfigure}
\begin{subfigure}{0.49\linewidth}
\centering
$\Sbin=\bmat{1 & 0 \\ 1 & 1}$
\end{subfigure}
\caption{DAG and associated sparsity matrix for the two-player problem.
\label{fig:ex_tpof}}
\end{figure}

This problem was solved in a prior version of this work~\cite{lessard_nayyar_tpof}, where the optimal controller was found to be of the form
\begin{align*}
\vec{u}_t^1 &= K^{11}_t \eec{\vec{x}_t}{\vec{i}_t^1}
&
\vec{u}_t^2 &= K^{21}_t \eec{\vec{x}_t}{\vec{i}_t^1} +
			  K^{22}_t \eecc{\vec{x}_t}{\vec{i}_t^{\{1,2\}}}
\end{align*}
In other words, both players must estimate the global state $\vec{x}_t$ using their available information. Additionally, the second player makes use of the estimate that the first player computed. This result agrees perfectly with Theorem~\ref{thm:main}. Note that here, Assumptions A1 and A2 are not restrictive. In other words, the structural result holds even when cost is coupled and noise is correlated. Infinite-horizon continuous-time versions of this problem were solved in~\cite{lessard_allerton} and a similar structure was found there as well. While these works only address the two-player version, it is clear how the structural result extends to a linear chain with $n$ nodes. For such a case, the optimal controller is of the form
\[
\vec{u}_t^k = \sum_{j=1}^k K_t^{kj} \eecc{\vec{x}_t}{\vec{i}_t^{\{1,\dots,j\}}}
\qquad\text{for }k=1,\dots,n
\]

\item (Broadcast cases) We define a broadcast DAG to be a graph with $n$ nodes, one of which is the hub. The $n=4$ cases are illustrated in Figure~\ref{fig:broadcast}.
\begin{figure}[ht]
\centering
\begin{subfigure}{0.24\linewidth}
\centering
\begin{tikzpicture}[thick,node distance=2.0cm,>=latex]
	\def\fs{\small}  
	\tikzstyle{block}=[circle,draw,fill=black!6,minimum height=2.1em]
	\def\h{1.1}
	\node [block](N1){$1$};
	\node [block] at (-\h,-\h) (N2){$2$};
	\node [block] at (0,-\h) (N3){$3$};
	\node [block] at (\h,-\h) (N4){$4$};
	\draw [->] (N1) -- (N2);
	\draw [->] (N1) -- (N3);
	\draw [->] (N1) -- (N4);
\end{tikzpicture}
\end{subfigure}
\begin{subfigure}{0.22\linewidth}
\centering
$\Sbin=\bmat{ 1 & 0 & 0 & 0 \\
			  1 & 1 & 0 & 0 \\
			  1 & 0 & 1 & 0 \\
			  1 & 0 & 0 & 1  }$
\end{subfigure}
\hfill
\begin{subfigure}{0.24\linewidth}
\centering
\begin{tikzpicture}[thick,node distance=2.0cm,>=latex]
	\def\fs{\small}  
	\tikzstyle{block}=[circle,draw,fill=black!6,minimum height=2.1em]
	\def\h{1.1}
	\node [block] at (-\h,0) (N1){$1$};
	\node [block] at (0,0) (N2){$2$};
	\node [block] at (\h,0) (N3){$3$};
	\node [block] at (0,-\h) (N4){$4$};
	\draw [->] (N1) -- (N4);
	\draw [->] (N2) -- (N4);
	\draw [->] (N3) -- (N4);
\end{tikzpicture}
\end{subfigure}
\begin{subfigure}{0.22\linewidth}
\centering
$\Sbin=\bmat{ 1 & 0 & 0 & 0 \\
			  0 & 1 & 0 & 0 \\
			  0 & 0 & 1 & 0 \\
			  1 & 1 & 1 & 1  }$
\end{subfigure}
\caption{DAG and associated sparsity matrix for the broadcast-out (left) and broadcast-in (right) architectures for $n=4$ nodes.
\label{fig:broadcast}}
\end{figure}
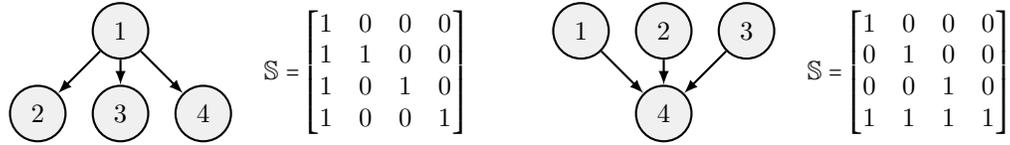

Note that all pairs of nodes have either a common ancestor or a common descendant, so Assumptions~A2 and~A2' are equivalent. Theorem~\ref{thm:main} provides the following.
\begin{enumerate}
\item In the broadcast-out case, we require that nodes $\{2,\dots,n\}$ have pairwise-decoupled cost. In this case, the optimal controller has the form
\begin{align*}
\vec{u}_t^1 &= K_t^{11} \eec{\vec{x}_t}{\vec{i}_t^1} & \\
\vec{u}_t^k &= K_t^{k1} \eec{\vec{x}_t}{\vec{i}_t^1} +
K_t^{kk} \eecc{\vec{x}_t^{\{1,k\}}}{\vec{i}_t^{\{1,k\}}}
& \text{for }k=2,\dots,n
\end{align*}
\item In the broadcast-in case, we require that the nodes $\{1,\dots,n-1\}$ have have pairwise-uncorrelated noise. In this case, the optimal controller has the form
\begin{align*}
\vec{u}_t^k &= K_t^{kk} \eecc{\vec{x}_t^{\{k,n\}}}{\vec{i}_t^k}
& \text{for }k=1,\dots,n-1 \\
\vec{u}_t^n &= \sum_{j=1}^{n-1} K_t^{nj} \eecc{\vec{x}_t^{\{j,n\}}}{\vec{i}_t^j}
+ K_t^{nn} \eecc{\vec{x}_t}{\vec{i}_t^{\{1,\dots,n\}}} &
\end{align*}
\end{enumerate}
A similar architecture was studied in \cite{lessard_broadcast} for the continuous-time infinite-horizon case, and a structure identical to the one above was found.

\item (Simple graph) Consider the simple five-node graph of Figure~\ref{fig:ex_simplegraph}. Because of Assumption~A2', we have the following restrictions:
\begin{itemize}
\item Nodes (3,4) and (4,5) have decoupled cost.
\item Nodes (1,2) have uncorrelated noise.
\item Nodes (1,4) have either decoupled cost or uncorrelated noise.
\end{itemize}
Theorem~\ref{thm:main} provides the following optimal controller structure, which we express using the formulation of Corollary~\ref{cor:main}.
\[
\bmat{ \vec{u}_t^1 \vphantom{i_t^{\{}} \\
       \vec{u}_t^2 \vphantom{i_t^{\{}} \\
       \vec{u}_t^3 \vphantom{i_t^{\{}} \\
       \vec{u}_t^4 \vphantom{i_t^{\{}} \\
       \vec{u}_t^5 }\vphantom{i_t^{\{} }  =
\underbrace{\bmat{K_t^{11} & 0 & 0 & 0 & 0 \vphantom{i_t^{\{}} \\
      0 & K_t^{22} & 0 & 0 & 0 \vphantom{i_t^{\{}} \\
      K_t^{31} & K_t^{32} & K_t^{33} & 0 & 0 \vphantom{i_t^{\{}} \\
      0 & K_t^{42} & 0 & K_t^{44} & 0 \vphantom{i_t^{\{}} \\
      K_t^{51} & K_t^{52} & K_t^{53} & 0 & K_t^{55} \vphantom{i_t^{\{}} }}_{K_t}
\bmat{ \eecc{\vec{x}_t^{\{1,3,5\}}}{\vec{i}_t^{1} } \\
	   \eecc{\vec{x}_t^{\{2,3,4,5\}}}{\vec{i}_t^{2} } \\
	   \eecc{\vec{x}_t^{\{1,2,3,5\}}}{\vec{i}_t^{\{1,2,3\}} } \\
	   \eecc{\vec{x}_t^{\{2,4\}}}{\vec{i}_t^{\{2,4\}} } \\
	   \eecc{\vec{x}_t^{\{1,2,3,5\}}}{\vec{i}_t^{\{1,2,3,5\}} } }
\]
Note that the matrix of control gains $K_t$ has a block-sparsity pattern that conforms to $\Sbin$, the sparsity matrix for this problem.

\end{list}

\section{Proof preliminaries}\label{sec:proof_prelim}

The main results of this paper,~Theorems~\ref{thm:main} and \ref{thm:main_2}, will be proved in the several sections that follow. In this section, we give an outline of the proof method and introduce required concepts that will be used later in the formal proof. Throughout this section and Section~\ref{sec:proof_main}, we will make assumptions A1 and A2 described in Section \ref{sec:assumptions}. 

\subsection{Proof outline}\label{sec:proof_outline}

Our proof technique may be thought of as a sequence of refinements that traverse the underlying DAG starting from the leaf nodes and finishing at the root nodes. To illustrate the process, consider the following four-node graph.

\begin{figure}[ht]
\centering
\begin{subfigure}{0.49\linewidth}
\centering
\begin{tikzpicture}[thick,node distance=2.0cm,>=latex]
	\def\fs{\small}  
	\tikzstyle{block}=[circle,draw,fill=black!6,minimum height=2.1em]
	\def\h{1.4}
	\def\r{0.48}
	\node [block] at (-\h,0) (N1){$1$};
	\node [block] at (0,0) (N2){$2$};
	\node [block] at (\h,\r) (N3){$3$};
	\node [block] at (\h,-\r) (N4){$4$};
	\draw [->] (N1) -- (N2);
	\draw [->] (N2) -- (N3);
	\draw [->] (N2) -- (N4);
\end{tikzpicture}
\end{subfigure}
\begin{subfigure}{0.49\linewidth}
\centering
$\Sbin=\bmat{ 1 & 0 & 0 & 0 \\
			  1 & 1 & 0 & 0 \\
			  1 & 1 & 1 & 0 \\
			  1 & 1 & 0 & 1  }$
\end{subfigure}
\caption{A four-node DAG and its associated sparsity pattern.
\label{fig:simple_4node}}
\end{figure}
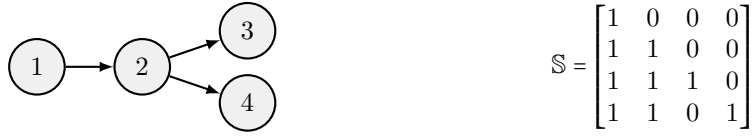
\noindent The most basic structural form of control strategies for the system illustrated in Figure~\ref{fig:simple_4node} is simply our initial information constraint~\eqref{eq:input_definition}, namely,
\begin{align*}
\vec{u}_t^1 &= f^1_t\bigl( \vec{i}_t^1 \bigr)	&
\vec{u}_t^2 &= f^2_t\bigl( \vec{i}_t^{\{1,2\}} \bigr)	&
\vec{u}_t^3 &= f^3_t\bigl( \vec{i}_t^{\{1,2,3\}} \bigr)	&
\vec{u}_t^4 &= f^4_t\bigl( \vec{i}_t^{\{1,2,4\}} \bigr)
\end{align*}
where the $f_t^i$ are linear functions. To prove Theorems~\ref{thm:main} and \ref{thm:main_2} for this graph, we can proceed as follows:

\paragraph{Step 1: leaf nodes.} Our first step examines the leaf nodes. We fix strategies of all nodes except a single leaf node, say node~4 in our current example. We consider node~4's centralized control problem and find a structural result for its optimal control strategy. We repeat the process for all leaf nodes. This step, described in detail in  Section~\ref{sec:P0},  yields the following structural result for optimal strategies in  the example of Figure~\ref{fig:simple_4node}:
\begin{align*}
\vec{u}_t^1 &= f^1_t\bigl( \vec{i}_t^1 \bigr)	&
\vec{u}_t^2 &= f^2_t\bigl( \vec{i}_t^{\{1,2\}} \bigr)	&
\vec{u}_t^3 &= h^3_t\bigl( \vec{i}_t^{\{1,2\}}, \vec{z}_t^{\anc{3}} \bigr)	&
\vec{u}_t^4 &= h^4_t\bigl( \vec{i}_t^{\{1,2\}}, \vec{z}_t^{\anc{4}} \bigr)
\end{align*}
where the $h_t^i$ are new linear functions.

\paragraph{Step 2: parents  of leaf nodes.} For the next step, we jointly consider the leaf nodes and their parent nodes (that is, the nodes whose strict descendants are leaf nodes). For the example of Figure~\ref{fig:simple_4node}, we consider nodes $\{2,3,4\}$ together. The strategies of these nodes depend on the common information $\vec{i}_t^{\{1,2\}}$ as well as the estimates computed in the previous step which are used only by the leaf nodes. 
We introduce a ``coordinator'' that knows the common information and selects the part of control actions that depends only on common information. The coordinator's problem is a centralized control problem for which we can find a structural result. This yields a refined structural result for optimal strategies:
 \begin{align*}
\vec{u}_t^1 &= f^1_t\bigl( \vec{i}_t^1 \bigr)	&
\vec{u}_t^2 &= m^2_t\bigl( \vec{i}_t^1, \vec{z}_t^\funnel{2} \bigr)	&
\vec{u}_t^3 &= m^3_t\bigl( \vec{i}_t^1, \vec{z}_t^\funnel{2}, \vec{z}_t^{\anc{3}} \bigr)	&
\vec{u}_t^4 &= m^4_t\bigl( \vec{i}_t^1, \vec{z}_t^\funnel{2}, \vec{z}_t^{\anc{4}} \bigr)
\end{align*}
where the $m_t^i$ are new linear functions.

\paragraph{Step 3: all other nodes.} The pattern is now clear. The next step is to look at nodes $\{1,2,3,4\}$ which together have $\vec{i}_t^1$ as the common information. By introducing a new coordinator that selects the part of control actions that depend only on this common information, we once again obtain a centralized problem that yields our final result.
\begin{align*}
\vec{u}_t^1 &= g^1_t\bigl( \vec{z}_t^\funnel{1} \bigr)	&
\vec{u}_t^2 &= g^2_t\bigl( \vec{z}_t^\funnel{1}, \vec{z}_t^\funnel{2} \bigr)	&
\vec{u}_t^3 &= g^3_t\bigl( \vec{z}_t^\funnel{1}, \vec{z}_t^\funnel{2}, \vec{z}_t^{\anc{3}} \bigr)	&
\vec{u}_t^4 &= g^4_t\bigl( \vec{z}_t^\funnel{1}, \vec{z}_t^\funnel{2}, \vec{z}_t^{\anc{4}} \bigr)
\end{align*}
where the $g_t^i$ are new linear functions. Note that in the above example although $\funnel{1}=\funnel{2}=\{1,2,3,4\}$, the estimates $\vec{z}_t^\funnel{1}$ and $\vec{z}_t^\funnel{2}$ are different. The latter is node~2's estimate, which uses~$\vec{i}_t^{\{1,2\}}$, while the former is node~1's estimate, which uses only~$\vec{i}_t^1$.

To summarize, we began with controllers that depend on an information history that grows with time (e.g. $\vec i_t^1$), and finished with controllers that only depend on conditional estimates (e.g. $\vec z_t^\funnel{1}$). These conditional estimates are \emph{sufficient statistics} that aggregate the information history. They only require finite storage, and may be computed recursively.
In Section~\ref{sec:Ps}, the process described above is formalized using induction, thus proving the structural result for any multitree.

\subsection{New definitions for multitrees}

New notation and concepts are required to make the leap from a simple graph like the example in Figure~\ref{fig:simple_4node} to a general multitree. A general multitree may have multiple root and leaf nodes, as well as multiple branches of varying lengths. We handle such cases by defining the concept of a \emph{generation}.

\begin{defn}\label{def:generations}
Suppose $G(\mathcal{V,E})$ is a multitree. Define \eemph{generations} recursively as follows.
\begin{align*}
\mathcal{G}^0 &= \sett{ i\in \mathcal{V} }{ \sdes{i} = \emptyset } \\
\mathcal{G}^k &= \sett{ i\in \mathcal{V}\setminus\textstyle{\bigcup_{j<k}}\mathcal{G}^{j} }{ \sdes{i} \subseteq \textstyle{\bigcup_{j<k}}\mathcal{G}^{j} }
\qquad \text{for $k=1,2,\dots$}
\end{align*}
\end{defn}
\noindent
For the example of Figure~\ref{fig:simple_4node}, we have:
\begin{align*}
\mathcal{G}^0 = \{3,4\} &&
\mathcal{G}^1 = \{2\} &&
\mathcal{G}^2 = \{1\} &&
\mathcal{G}^3 = \mathcal{G}^4 = \dots = \emptyset
\end{align*}
For the example of Figure~\ref{fig:ex_simplegraph}, we have:
\begin{align*}
\mathcal{G}^0 = \{4,5\} &&
\mathcal{G}^1 = \{3\} &&
\mathcal{G}^2 = \{1,2\} &&
\mathcal{G}^3 = \mathcal{G}^4 = \dots = \emptyset
\end{align*}
The nodes in $\mathcal{G}^0$ are precisely the leaf nodes. It is clear that the $\mathcal{G}^k$ sets are a partition of~$\mathcal{V}$. Furthermore, for some $m \leq n$, the first $m$ generations  are nonempty and all subsequent ones are empty. For convenience, we use the notation
\begin{align*}
\mathcal{G}^{\leq m} &:= \bigcup_{i\le m} \mathcal{G}^i &
&\text{and}&
\mathcal{G}^{\geq m} &:= \bigcup_{i \ge m} \mathcal{G}^i
\end{align*} 
The proof of our main results will make use of a graph traversal that proceeds generation-by-generation rather than node-by-node as we did in the example of Figure~\ref{fig:simple_4node}. We will also require some new definitions in order to capture the features of more complicated multitrees.

\begin{defn}[siblings, co-parents and non-relatives] \label{def:sibs}
Suppose $G(\mathcal{V,E})$ is a multitree. Given  $j \in \mathcal{V}$, define the following subsets of $\mathcal{V}$.
\begin{enumerate}[(i)]
\item The \eemph{siblings} of $j$:
	$\sibling{j} = \bigcup_{i \in \anc{j}}\des{i}\setminus \funnel{j}$.
\item The \eemph{co-parents} of $j$:
	$\coparent{j} = \bigcup_{i \in \des{j}}\anc{i}\setminus \funnel{j}$.
\item The \eemph{non-relatives} of $j$:
	$\nonrelative{j} = \mathcal{V}\setminus(\funnel{j}\cup\coparent{j}\cup\sibling{j})$
\end{enumerate}
\end{defn}
In other words, the siblings of~$j$ are the nodes that are neither ancestors nor descendants of~$j$, but that share a common ancestor with~$j$. Similarly, the co-parents of~$j$ are the nodes that are neither ancestors nor descendants of~$j$, but that share a common descendant with~$j$. The non-relatives of~$j$ are all the remaining nodes once we have excluded descendants, ancestors, siblings, and co-parents of~$j$. We now observe that $\mathcal{V}$ may be partitioned using Definition~\ref{def:sibs}.
\begin{lem}\label{lem:disjoint}
Suppose $G(\mathcal{V},\mathcal{E})$ is a multitree. For every $j \in \mathcal{V}$, the six sets $\{j\}$, $\sanc{j}$, $\sdes{j}$, $\sibling{j}$, $\coparent{j}$, $\nonrelative{j}$ are a partition of $\mathcal{V}$.
\end{lem}
\begin{proof}
Recall from Section~\ref{sec:prelim} that $\{j\}$, $\sanc{j}$, and $\sdes{j}$ form a partition of $\funnel{j}$. Next, we show that $\sibling{j}$ and $\coparent{j}$ are disjoint. Suppose to the contrary that there exists some $k \in \sibling{j}\cap\coparent{j}$. Then (i) $j$ and $k$ have both a common ancestor~$a$ and a common descendant~$d$ and (ii)  $k \notin \funnel{j}$,  it follows that $k\nleftrightarrow j$. Therefore $\{a,j,k,d\}$ form a diamond, which violates the multitree assumption. Therefore,  $\sibling{j}\cap\coparent{j} = \emptyset$. By Definition~\ref{def:sibs}, $\nonrelative{j}$ completes the partition, as required.
\end{proof}

\subsection{Aggregated graph and dynamics}\label{sec:aggregated}

A key ingredient mentioned in the proof outline of Section~\ref{sec:proof_outline} is the identification of centralized subproblems hidden within decentralized architectures. To make such identifications more transparent in the proof to come, we define aggregated representations of $G(\mathcal{V,E})$ that highlight the structures that are relevant for the various decision-makers. 

\begin{defn}[six-node aggregated graph]
Suppose $G(\mathcal{V},\mathcal{E})$ is a multitree. For  every $j \in \mathcal{V}$, we define the \eemph{six-node aggregated graph} centered at~$j$, denoted $G^j(\mathcal{V}^j,\mathcal{E}^j)$, as the multitree of Figure~\ref{fig:aggregate}.
\begin{figure}[ht]
\centering
\begin{tikzpicture}[thick,>=latex]
	\tikzstyle{block}=[circle,draw,fill=black!6,minimum height=2.4em];
	\tikzstyle{cblock}=[cloud,draw,fill=black!6,cloud puffs=7.5,cloud puff arc=120, aspect=1.7, inner sep=1mm];
	\def\x{2.7};
	\def\y{1.1};
	\node [cblock](N1) at (-0.5*\x,\y) {$\textbf{1}:\sanc{j}$};
	\node [cblock](N2) at (0.5*\x,\y) {$\textbf{2}:\coparent{j}$};
	\node [block](N3) at (-1.2*\x,0) {$\textbf{3}:j$};
	\node [cblock](N4) at (1.2*\x,0) {$\textbf{4}:\nonrelative{j}\!$};
	\node [cblock](N5) at (-0.5*\x,-\y) {$\textbf{5}:\sdes{j}$};
    \node [cblock](N6) at (0.5*\x,-\y) {$\textbf{6}:\sibling{j}$};
	\draw [->] (N1) -- (N3);
	\draw [->] (N1) -- (N6);
	\draw [->] (N2) -- (N5);
    \draw [->] (N3) -- (N5);
	\draw [->] (N2) -- (N6);
	\draw [->] (N2) -- (N4);
	\draw [->] (N4) -- (N6);
	\draw [->] (N1) -- (N5);
\end{tikzpicture}
\caption{Diagram of $G^j(\mathcal{V}^j,\mathcal{E}^j)$, the six-node aggregate graph centered at~$j$.\label{fig:aggregate}}
\end{figure}
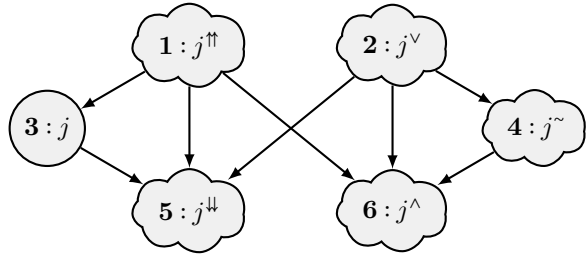

\end{defn}
The cloud-shaped nodes in $G^j$ indicate the aggregation of several nodes from the original graph $G(\mathcal{V},\mathcal{E})$. Thus, a node $s$ in  $G^j$ represents a subset $s \subseteq \mathcal{V}$. A directed edge $(s_1,s_2)\in\mathcal{E}^j$ means that the multitree structure of the original graph $G(\mathcal{V},\mathcal{E})$ permits the existence of nodes  $i \in s_1$, $k \in s_2$ with $(i,k) \in \mathcal{E}$.  On the other hand, $(s_1,s_2)\notin\mathcal{E}^j$ if and only if the multitree nature of the original graph rules out the existence of nodes $i \in s_1$, $k \in s_2$ with $(i,k) \in \mathcal{E}$.  
For example, consider the nodes $\coparent{j}$ and $\sibling{j}$ in the aggregated graph. Since there is an edge $\coparent{j} \to \sibling{j}$ but no edge $\sibling{j} \to \coparent{j}$, we conclude that there may be nodes $i\in\coparent{j}$ and $k\in\sibling{j}$ with $(i,k) \in \mathcal{E}$, but we will never have $(k,i) \in \mathcal{E}$. This  relationship between $G^j(\mathcal{V}^j,\mathcal{E}^j)$ and $G(\mathcal{V},\mathcal{E})$ is proved in the following lemma.

\begin{lem}
Suppose $G(\mathcal{V,E})$ is a multitree. For any $j\in\mathcal{V}$, let $G^j(\mathcal{V}^j,\mathcal{E}^j)$ be the corresponding six-node aggregated graph centered at $j$.
For all $s_1,s_2\in\mathcal{V}^j$, if $(s_1,s_2)\notin\mathcal{E}^j$ then there do not exist $i \in s_1$, $k \in s_2$ with $(i,k) \in \mathcal{E}$.
\end{lem}
\begin{proof}
The lemma is essentially a consequence of the multitree nature of $G(\mathcal{V},\mathcal{E})$ and the definition of $G^j(\mathcal{V}^j,\mathcal{E}^j)$. For example, consider $s_1 =\sanc{j}, s_2=\coparent{j}$. There is no edge connecting these nodes in $G^j$. Now, suppose $i,k$ are nodes in the original graph $G$ with  $i\in \sanc{j}$ and $k\in\coparent{j}$. We cannot have $(i,k) \in \mathcal{E}$, for this would create a diamond $\{i,j,k,d\}$, where $d$ is the common descendant of $j$ and $k$. We cannot have $(k,i) \in \mathcal{E}$ either, since then $k$ would be an ancestor of $j$, and then $k$ should belong to $\sanc{j}$, a contradiction.
Similar arguments can be made with all other pairs of nodes to establish the result of the lemma.
\end{proof} 

The aggregated graph centered at node $j$ will serve as the basis for aggregating states and dynamics of the overall system when we want to focus on node $j$. 
We define aggregated states $\set{\bar{\mathbf{x}}^m}{m=1,\dots,6}$ corresponding to each node in $G^j(\mathcal{V}^j,\mathcal{E}^j)$ as follows:
\begin{align*}
\bar{\mathbf{x}}^1 &= \mathbf{x}^\sanc{j} &
\bar{\mathbf{x}}^2 &= \mathbf{x}^\coparent{j} &
\bar{\mathbf{x}}^3 &= \mathbf{x}^j &
\bar{\mathbf{x}}^4 &= \mathbf{x}^\nonrelative{j} &
\bar{\mathbf{x}}^5 &= \mathbf{x}^\sdes{j} &
\bar{\mathbf{x}}^6 &= \mathbf{x}^\sibling{j}
\end{align*}
and define $\bar{\mathbf{u}}$, $\bar{\mathbf{y}}$, $\bar{\mathbf{w}}$, and $\bar{\mathbf{v}}$ in a similar fashion. Once this is done, the aggregated dynamics centered at $j$ may be written compactly as
\begin{equation}\label{eq:state_eqns_aggregated}
\begin{aligned}
\bar{\mathbf{x}}_+ &\eqt \bar A\bar{\mathbf{x}} + \bar B\bar{\mathbf{u}} + \bar{\mathbf{w}} \\
\bar{\mathbf{y}} &\eqt \bar C\bar{\mathbf{x}} + \bar{\mathbf{v}}
\end{aligned}
\end{equation}
Note that all vectors and matrices in ~\eqref{eq:state_eqns_aggregated} are  merely rearrangements of corresponding vectors and matrices in~\eqref{eq:state_eqns}. However, the new expression of~\eqref{eq:state_eqns_aggregated} has an important structural property. Namely, if we split $\bar A$, $\bar B$, and $\bar C$ into blocks according to the six states of the aggregated system, they will have a sparsity pattern that conforms to the graph of Figure~\ref{fig:aggregate}. We will show that despite being a coarser representation of the system dynamics, Figure~\ref{fig:aggregate} captures the structure that is \emph{relevant to node}~$j$ for the purpose of determining optimal decisions.

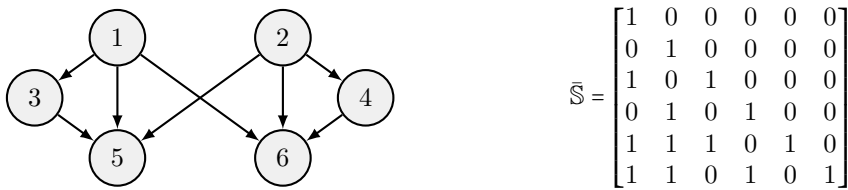
\begin{figure}[ht]
\centering
\begin{subfigure}{0.49\linewidth}
\centering
\begin{tikzpicture}[thick,>=latex]
	\tikzstyle{block}=[circle,draw,fill=black!6,minimum height=2.1em]
	\def\x{2.2};
	\def\y{0.8};
	\node [block](N1) at (0.5*\x,\y) {$1$};
	\node [block](N2) at (1.5*\x,\y) {$2$};
	\node [block](N3) at (0,0) {$3$};
	\node [block](N4) at (2.0*\x,0) {$4$};
	\node [block](N5) at (0.5*\x,-\y) {$5$};
	\node [block](N6) at (1.5*\x,-\y) {$6$};
	\draw [->] (N1) -- (N3);
	\draw [->] (N1) -- (N6);
	\draw [->] (N2) -- (N5);
    \draw [->] (N3) -- (N5);
	\draw [->] (N2) -- (N6);
	\draw [->] (N2) -- (N4);
	\draw [->] (N4) -- (N6);
	\draw [->] (N1) -- (N5);
\end{tikzpicture}
\end{subfigure}
\begin{subfigure}{0.49\linewidth}
\centering
$\bar\Sbin = \bmat{ 1 & 0 & 0 & 0 & 0 & 0 \\
			0 & 1 & 0 & 0 & 0 & 0 \\
			1 & 0 & 1 & 0 & 0 & 0 \\
			0 & 1 & 0 & 1 & 0 & 0 \\
			1 & 1 & 1 & 0 & 1 & 0 \\
			1 & 1 & 0 & 1 & 0 & 1 }$
\end{subfigure}
\caption{Six-node aggregated DAG and associated sparsity pattern.
\label{fig:six_node_centralized}}
\end{figure}

\subsection{The six-node centralized problem}
If we adopt the aggregated representation of systems dynamics based on the aggregated graph centered at $j$, the overall multitree ``looks like'' the six-node graph of Figure~\ref{fig:six_node_centralized}. The six-node graph and the corresponding
 sparsity pattern $\bar\Sbin$ given in Figure~\ref{fig:six_node_centralized} are universal in the sense that we get the same sparsity pattern  regardless of which $j\in\mathcal{V}$ is chosen to center the aggregation $G^j$.  Of course, some nodes in the aggregated graph may be empty. For example, if $j$ is a leaf node, then nodes $\sdes{j}, \coparent{j}$ are empty.
 
Motivated by the universality of the graph in Figure~\ref{fig:six_node_centralized}, we will now investigate an instance of Problem~1 corresponding to this graph. We will however assume that node $3$ is the only decision-maker, and that all observation processes other than the measurement at node $3$ is identically zero.


Therefore, we may write the dynamics of the six-node problem as
\begin{equation}\label{eq:six_node_eqns_structure}
\begin{aligned}
\bmat{ \vec{x}_+^1 \\ \vec{x}_+^2 \\
       \vec{x}_+^3 \\ \vec{x}_+^4 \\
       \vec{x}_+^5 \\ \vec{x}_+^6 }
&\eqt
\bmat{ A_{11} & 0      & 0      & 0      & 0      & 0     \\
	   0      & A_{22} & 0      & 0      & 0      & 0     \\
	   A_{31} & 0      & A_{33} & 0      & 0      & 0     \\
	   0      & A_{42} & 0      & A_{44} & 0      & 0     \\
	   A_{51} & A_{52} & A_{53} & 0      & A_{55} & 0     \\
	   A_{61} & A_{62} & 0      & A_{64} & 0      & A_{66} } 
\bmat{ \vec{x}^1 \\ \vec{x}^2 \\
       \vec{x}^3 \\ \vec{x}^4 \\
       \vec{x}^5 \\ \vec{x}^6 }
+
\bmat{  0      \\
	    0      \\
	    B_{33} \\
	    0      \\
	    B_{53} \\
	    0      } \vec{u}_3
+
\bmat{ \vec{w}^1 \\ \vec{w}^2 \\
       \vec{w}^3 \\ \vec{w}^4 \\
       \vec{w}^5 \\ \vec{w}^6 } \\
\vec{y}^3 &\eqt 
C_{31} \vec{x}^1 + C_{33} \vec{x}^3 + \vec{v}^3
\end{aligned}
\end{equation}
Controller~$3$ selects its actions according to a control strategy
$f^3 :=(f^3_0,\ldots,f^3_{T-1})$ of the form
\begin{equation}\label{eq:six_node_input_definition}
	\mathbf{u}_t^3 = f^3_t(\,\vec{i}^3_t\,) = f^3_t(\,\vec{y}^3_{1:t-1},\vec{u}^3_{1:t-1}\,) 
\quad\text{for }0\le t\le T-1
\end{equation}
The performance of control strategy $f^3$ is measured by the finite horizon expected quadratic cost given by
\begin{equation}\label{eq:cost_centralized}
\hat{\mathcal{J}}_0(f^3) =
\ee^{f^3} \biggl( \sum_t
\bmat{\mathbf{x}\\\mathbf{u}}^\tp \bmat{Q & S \\ S^\tp & R} \bmat{\mathbf{x}\\ \mathbf{u}}
+ \mathbf{x}_T^\tp P_\textup{final} \mathbf{x}_T \biggr)
\end{equation}
The expectation is taken with respect to the joint probability measure on $(\mathbf{x}_{0:T},  \mathbf{u}_{0:T-1})$ induced by the choice
of $f^3$.  The six-node problem of interest is as follows.

\boxit{
\begin{problem}[Six-node centralized LQG]
  \label{prob:CLQG_6node}
  For the model of~\eqref{eq:indep_noise}--\eqref{eq:initial_conditions} and \eqref{eq:six_node_eqns_structure}--\eqref{eq:six_node_input_definition}, and subject to Assumption  A2, find a control strategy $f^3$ that minimizes the cost~\eqref{eq:cost_centralized}.
\end{problem}
}

Problem~\ref{prob:CLQG_6node} is clearly a  centralized control problem because there is a single decision-maker. The classical LQG result \cite{kumar_varaiya} states that the optimal control law is of the form $\vec{u}_t^3 = h_t^3(\vec{z}_t)$ where $h_t^3$ is a linear function, and $\vec{z}_t = \eecs{\vec{x}_t}{\vec{i}_t^3}$ is the  estimate of the global state. The additional structure imposed in~\eqref{eq:six_node_eqns_structure} together with Assumption~A2 allows us to refine the classical result. In particular, it suffices to estimate $\vec{x}_t^1$, $\vec{x}_t^3$, and $\vec{x}_t^5$.

\begin{thm}\label{thm:six_node_result}
The optimal control strategy in Problem \ref{prob:CLQG_6node} is of the form
\begin{equation}
\mathbf{u}^3_t = h^3_t(\mathbf{z}^{\funnel{3}}_t)
\end{equation}
where $\vec{z}^{\funnel{3}}_t = (\vec{z}^1_t, \vec{z}^3_t, \vec{z}^5_t)= \eecs{\VEC{x}_t^{\funnel{3}}}{\vec{i}^{3}_t}$ and $h^3_t$ is a linear function. Further, $\vec{z}^{\funnel{3}}$ has a linear update equation of the form
\begin{equation}\label{eq:est_centralized}
\vec z_+^\funnel{3} \eqt
\bmat{A_{11} & 0 & 0 \\ A_{31} & A_{33} & 0 \\ A_{51} & A_{53} & A_{55}}\vec z^\funnel{3} + \bmat{0 \\ B_{33} \\ B_{53}}\vec{u}^3 -  L \bigl( \vec{y}^3- \bmat{C_{31}& C_{33}} \vec z^\anc{3} \bigr)
\end{equation}
where the matrix gain $ L$ does not depend on the choice of $h^3$.
\end{thm}

\begin{proof}
 See Appendix \ref{sec:proof_thm3}.
\end{proof}
Note that the optimal strategies $h^3_t$ as well as the gains $L_t$ can be explicitly computed by algebraic Riccati recursions. This fact will not be needed in this paper, but we will make use of it in Part~II.


\subsection{A partial separation result for the aggregated graph}
Our final preliminary step before the formal proof is to establish a \emph{partial separation result} for Problem~\ref{prob:NPLQG}. The classical separation result for  centralized LQG control states that a controller's posterior belief on the state is independent of its control strategy  \cite{kumar_varaiya}.  In particular,  for any two (linear) control strategies $f$ and $g$ in the centralized LQG problem, the conditional mean and the conditional covariance matrix are the same:
\begin{align}
&\ee^{f}\left(\vec{x}_t\middle\vert  y_{0:t-1}, u_{0:t-1}\right) = \ee^{g}\left(\vec{x}_t\middle\vert  y_{0:t-1}, u_{0:t-1}\right), \\
&\ee^{f}\left((\vec{x}_t - \hat{x}_t)(\vec{x}_t - \hat{x}_t)^\tp\middle\vert  y_{0:t-1}, u_{0:t-1}\right) = \ee^{g}\left((\vec{x}_t - \hat{x}_t)(\vec{x}_t - \hat{x}_t)^\tp\middle\vert  y_{0:t-1}, u_{0:t-1}\right),
\end{align} 
where $\hat{x}_t = \ee^{f}\left(\vec{x}_t\middle\vert  y_{0:t-1}, u_{0:t-1}\right) = \ee^{g}\left(\vec{x}_t\middle\vert  y_{0:t-1}, u_{0:t-1}\right)$.

This \emph{complete separation} of estimation from control strategies does not hold, in general, for our problem. The following lemma shows how estimation at node $j$ can be \emph{separated from some parts of the control strategy profile.}
\begin{lem}\label{lem:estimation}
Consider any node $j$. Consider any fixed choice $f = (f^1,f^2,\ldots,f^n)$ of linear control strategies of all $n$ controllers. Then, 
\begin{enumerate}
\item  $\ee^{{f}}\left(\vec{x}^{\coparent{j}}_t,\vec{i}^{\coparent{j}}_t\middle\vert\vec i^{\anc{j}}_t\right) =0$ and  $\ee^{f}\left(\vec{x}^{\nonrelative{j}}_t,\vec{i}^{\nonrelative{j}}_t\middle\vert\vec{i}^{\anc{j}}_t\right) =0$.

\item Group the control strategies of the $n$ nodes according to the aggregated graph centered at $j$, that is, write $f$ as $f=(f^{\sanc{j}}, f^{\coparent{j}}, f^j, f^{\nonrelative{j}}, f^{\sdes{j}}, f^{\sibling{j}})$.  Let $g$ be another linear strategy profile given as $g = (g^{\sanc{j}}, g^{\coparent{j}}, g^j, g^{\nonrelative{j}}, f^{\sdes{j}}, g^{\sibling{j}})$. Note that the descendants of node $j$ have the same control strategies under $f$ and $g$. Then,
\[\ee^{f}\left(\vec{x}^{\funnel{j}}_t,\vec{i}^{\funnel{j}}_t\middle\vert\vec i^{\anc{j}}_t = i^{\anc{j}}_t\right) = \ee^{g}\left(\vec{x}^{\funnel{j}}_t,\vec{i}^{\funnel{j}}_t\middle\vert\vec i^{\anc{j}}_t = i^{\anc{j}}_t\right).\]

\item Further,   let $ g = (g^{\sanc{j}}, g^{\coparent{j}}, g^j, g^{\nonrelative{j}}, g^{\sdes{j}}, g^{\sibling{j}})$ be any linear control strategy that satisfies
${g}^i_t(\vec{i}^{\anc{i}}_t) = {f}^i_t(\vec{i}^{\anc{i}}_t)  + m^i_t(\vec{i}^{\coparent{j}}_t)$,  
for all $i \in \sdes{j}$. Then,
\[\ee^{f}\left(\vec{x}^{\funnel{j}}_t,\vec{i}^{\funnel{j}}_t\middle\vert\vec i^{\anc{j}}_t = i^{\anc{j}}_t\right) = \ee^{g}\left(\vec{x}^{\funnel{j}}_t,\vec{i}^{\funnel{j}}_t\middle\vert\vec i^{\anc{j}}_t = i^{\anc{j}}_t\right).\]
\item Let $ g = (g^{\sanc{j}}, g^{\coparent{j}}, g^j, g^{\nonrelative{j}}, g^{\sdes{j}}, g^{\sibling{j}})$ be any linear control strategy that satisfies
\[ g^i_t(\vec{i}^{\anc{i}}_t) = f^i_t(\vec{i}^{\anc{i}}_t)   
+ \ell^i_t(\vec{i}^{\anc{j}}_t), ~~\mbox{ for all $i \in \sdes{j}$}.\]
 Then, the conditional  covariance matrix of  $(\vec{x}^{\funnel{j}}_t,\vec{i}^{\funnel{j}}_t, \vec{y}^{\anc{j}}_t )$ given $\vec{i}^{\anc{j}}_t$ is the same under $f$ and $g$.

\item Let $ h = (h^{\sanc{j}}, h^{\coparent{j}}, h^j, h^{\nonrelative{j}}, h^{\sdes{j}}, h^{\sibling{j}})$ be any linear control strategy that satisfies
\[{h}^i_t(\vec{i}^{\anc{i}}_t) = {f}^i_t(\vec{i}^{\anc{i}}_t)   + m^i_t(\vec{i}^{\coparent{j}}_t)
+ \ell^i_t(\vec{i}^{\anc{j}}_t), ~~\mbox{for all $i \in \sdes{j}$}.\] 
Then, the conditional  covariance matrix of $\vec{y}^{\anc{j}}_t$ given $\vec{i}^{\anc{j}}_t$ as well as the conditional cross covariance between $(\vec{x}^{\funnel{j}}_t,\vec{i}^{\funnel{j}}_t)$ and  $\vec{y}^{\anc{j}}_t $ given $\vec{i}^{\anc{j}}_t$ are the same under $f$ and $h$.
 \end{enumerate} 
\end{lem}
\begin{proof}
See Appendix \ref{sec:proof_estimation}.
\end{proof}

\section{Proof of main results} \label{sec:proof_main}

As explained in Section~\ref{sec:proof_outline}, we prove the structural results of Theorems~\ref{thm:main} and~\ref{thm:main_2} by traversing the graph from leaf nodes to root nodes. The proof uses mathematical induction, so we begin by stating the induction hypothesis $P(s)$. Recall that $\mathcal{G}^{\leq s}$ is the union of $\mathcal{G}^0, \mathcal{G}^1,\ldots,\mathcal{G}^s$ (see Definition~\ref{def:generations}).

\begin{prop*}~
\begin{enumerate}
\item  There is no loss in optimality in jointly restricting all nodes $j \in \mathcal{G}^{\leq s}$ to strategies of the form
\begin{align}
  \vec u^j_t &= \sum_{a \in \anc{j}\cap\mathcal{G}^{\leq s}}g^{ja}_t(\vec{z}^{\funnel{a}}_t) +\sum_{b \in \sanc{j}\cap\mathcal{G}^{\geq s+1}} h_t^{jb}(\vec i^b_t) \label{eq:induction_hypothesis2}
\end{align}
where $g_t^{ja}(\cdot)$ and $h_t^{jb}(\cdot)$ are linear functions and $\vec{z}^{\funnel{a}}_t = \eecc{\vec{x}^{\funnel{a}}_t}{\vec{i}^{\anc{a}}_t}$.

\item For each $j \in \mathcal{G}^{\leq s}$, the dynamics of $\vec{z}^{\funnel{j}}_t = \eecc{\vec{x}^{\funnel{j}}_t}{\vec{i}^{\anc{j}}_t}$ are given by
\begin{equation}\label{eq:z_Gs_dynamics}
\begin{aligned}
\vec{z}^{\funnel{j}}_0 &= 0 \\
\vec{z}^{\funnel{j}}_+ &\eqt
A^{\funnel{j}\funnel{j}}\vec{z}^{\funnel{j}}
 + B^{\funnel{j}\funnel{j}}\bmat{ \vec{u}^{\anc{j}} \\ \{\vec{\hat u}^{ij}\}_{i \in \sdes{j}}} - {L}^j\left(\vec{y}^{\anc{j}} -C^{\anc{j}\anc{j}}\vec{z}^{\anc{j}}\right)
\end{aligned}
\end{equation}
where for each $i\in\sdes{j}$, we have defined
\begin{equation} \label{eq:utilde_0} \vec{\hat u}^{ij}_t \defeq  \sum_{a \in \anc{j}\cap\mathcal{G}^{\leq s}}g^{ia}_t(\vec{z}^{\funnel{a}}_t) + \sum_{b \in \anc{i} \cap \sdes{j}} g^{ib}_t(E^{b,j}\vec{z}^\funnel{j}_t)  +\sum_{c \in \sanc{j}\cap\mathcal{G}^{\geq s+1}} h^{ic}_t (\vec i^c_t),
\end{equation}
and the matrices $E^{b,j}$ in \eqref{eq:utilde_0} are the matrices defined in Defintion~\ref{def:Gmatrix}. 

The matrices ${L}^j_t$ in \eqref{eq:z_Gs_dynamics} depend only on $\{g^{ib}: i\in \sdes{j}, b \in \anc{i} \cap \sdes{j}\}$.
\end{enumerate}
\end{prop*}

\begin{rem}
For linear control strategies described by \eqref{eq:induction_hypothesis2}, it can be established that   $\vec{\hat u}^{ij}_t$ of \eqref{eq:utilde_0} is in fact $\eecc{\vec{u}^i_t}{\vec{i}^{\anc{j}}_t}$. This interpretation of $\vec{\hat u}^{ij}_t$ as controller $j$'s estimate of controller $i$'s action will be helpful in the following proofs.
\end{rem}

Note that for $s=n$, $\mathcal{G}^{\le s} = \mathcal{V}$, $\mathcal{G}^{\ge s+1} = \emptyset$. Therefore, if we can show that Proposition $P(s)$ holds for $s=n$, then   we can directly obtain Theorems~\ref{thm:main} and~\ref{thm:main_2} from this proposition by using matrices $K^{ij}$ to represent the linear functions $g^{ij}$. 

Our argument for proving $P(0), P(1),\ldots,P(n)$ should be viewed as successive refinement of structural result for optimal control strategies.  Let $\mathcal{C}$ denote the set of all linear strategy profiles and $\mathcal{C}^s$, $s=0,1,\ldots,n$ be the set of strategy profiles that have the structural form required by Proposition $P(s)$. Then, $\mathcal{C} \supset \mathcal{C}^0 \supset \dots \supset \mathcal{C}^n$. We start with any arbitrary strategy profile in $\mathcal{C}$. The essence of our proof for Proposition $P(0)$ is that for any arbitrary choice of strategy profile in $\mathcal{C}$, there always exists another strategy profile in $\mathcal{C}^0$ 
with better or equal performance.
 Therefore, we can restrict attention to strategy profiles in $\mathcal{C}^0$ without compromising performance. Next, we consider any arbitrary strategy profile in $\mathcal{C}^0$. We will prove that for any such choice from $\mathcal{C}^0$, there always exists another strategy profile in $\mathcal{C}^1$ 
with better or equal performance. Therefore, we can restrict attention to strategy profiles in $\mathcal{C}^1$ without compromising performance. We continue this argument inductively.

\subsection{Leaf nodes: proof of $P(0)$}\label{sec:P0}

In the base case $P(0)$, we consider nodes $j \in \mathcal{G}^0$, which are the leaf nodes of the graph. This case is much simpler than the general one, because we have $\sdes{j}=\emptyset$ and $\funnel{j}=\anc{j}$. The base case $P(0)$ can be written as

\begin{enumerate}
\item  There is no loss in optimality in jointly restricting all $j \in \mathcal{G}^{0}$ to strategies of the form
\begin{equation}\label{pr:leaf1}
  \vec u^j_t = g^{jj}_t(\vec{z}^{\anc{j}}_t) +\sum_{b \in \sanc{j}} h_t^{jb}(\vec i^b_t)
\end{equation}
where $g_t^{jj}(\cdot)$ and $h_t^{jb}(\cdot)$ are linear functions.

\item For $j \in \mathcal{G}^{0}$, the dynamics of $\vec{z}^{\anc{j}}_t = \eecc{\vec{x}^{\anc{j}}_t}{\vec{i}^{\anc{j}}_t}$ are given by
\begin{equation}\label{pr:leaf2}
\begin{aligned}
\vec{z}^{\anc{j}}_0 &= 0 \\
\vec{z}^{\anc{j}}_+ &\eqt
A^{\anc{j}\anc{j}}\vec{z}^{\funnel{j}}
 + B^{\anc{j}\anc{j}} \vec{u}^{\anc{j}} - {L}^j\bigl(\vec{y}^{\anc{j}} -C^{\anc{j}\anc{j}}\vec{z}^{\anc{j}}\bigr)
\end{aligned}
\end{equation}
and the matrices ${L}^j_t$ in \eqref{pr:leaf2} do not depend on the choice of strategy profile.
\end{enumerate}

To prove the above statements, consider any node $j\in\mathcal{G}^0$ and fix arbitrary linear control strategies for all nodes except node $j$. We will consider the problem of finding the best control strategy for node $j$ in response to the arbitrary choice of linear control strategies of all other controllers. We will consider the  aggregated graph centered at $j$ and argue that controller $j$'s problem can be viewed as an instance of the centralized Problem~\ref{prob:CLQG_6node}.

\begin{lem}\label{lem:leaf_sixnode}
In Problem~\ref{prob:NPLQG}, pick any $j \in \mathcal{G}^0$ and  any fixed linear strategies for all nodes $i\neq j$.
Then controller $j$'s optimization problem is an instance of the six-node centralized problem with states, measurements, and inputs given by
\begin{align} \label{eq:leaf_to_6_node}
\begin{aligned} 
\bar{\vec{x}}^1_t &=
	\bmat{\vec{x}^{\sanc{j}}_t \\ \vec{i}^{\sanc{j}}_t}  & 
\bar{\vec{x}}^2_t &= \emptyset  &
\bar{\vec{x}}^3_t &= \vec{x}^j_t \\
\bar{\vec{x}}^4_t &=
	\bmat{\vec{x}^{\nonrelative{j}}_t \\ \vec{i}_t^{\nonrelative{j}}} &
\bar{\vec{x}}^5_t &= \emptyset &
\bar{\vec{x}}^6_t &= \bmat{\vec{x}^{\sibling{j}}_t \\ \vec{i}_t^{\sibling{j}}}
\end{aligned}
& &
\bar{\vec{y}}_t^3 &= \bmat{\vec{y}_t^j \\ \vec{y}_t^\sanc{j} \\ \vec{u}_t^\sanc{j}} &
\bar{\vec{u}}_t^3 &= \vec{u}_t^j
\end{align}
\end{lem}

\begin{proof}
See Appendix \ref{sec:leaf_sixnode}.
\end{proof}

Because of Lemma~\ref{lem:leaf_sixnode} , we can apply the result of Theorem~\ref{thm:six_node_result} to controller $j$'s problem. Therefore, the optimal $\vec{u}_t^j$ is a linear function of $\eecs{\bar{\vec x}^1_t,\bar{\vec x}^3_t}{\vec{i}^{\anc{j}}_t} = \eecs{\vec{x}^{\sanc{j}}_t, \vec{i}^{\sanc{j}}_t, \vec{x}^j_t}{\vec{i}^{\anc{j}}_t}$.  Since $\vec{i}^{\sanc{j}}_t \subset \vec{i}^{\anc{j}}_t$, $\eecs{\vec{i}^{\sanc{j}}_t}{\vec{i}^{\anc{j}}_t} = \vec{i}^{\sanc{j}}_t$. Therefore, any linear function of $\eecs{\vec{x}^{\sanc{j}}_t, \vec{i}^{\sanc{j}}_t, \vec{x}^j_t}{\vec{i}^{\anc{j}}_t}$ can be written in the form of \eqref{pr:leaf1} .

Further, since the $(\vec{x}^{\anc{j}}_t, \vec{y}^{\anc{j}}_t)$ dynamics are of the form
\begin{align*}
\vec{x}^{\anc{j}}_+ &\eqt A^{\anc{j}\anc{j}}\vec{x}^{\anc{j}} + B^{\anc{j}\anc{j}}\vec{u}^{\anc{j}}_t + \vec{w}^{\anc{j}} \\
\vec{y}^{\anc{j}} &\eqt C^{\anc{j}\anc{j}}\vec{x}^{\anc{j}}  + \vec{v}^{\anc{j}}
\end{align*}
and $\vec{u}^{\anc{j}}_t$ is a function of $\vec{y}^{\anc{j}}_{0:t-1},\vec{u}^{\anc{j}}_{0:t-1}$, it follows that the estimate $\vec{z}^{\anc{j}}_t = \eecs{\vec{x}^{\anc{j}}_t}{\vec{y}^{\anc{j}}_{0:t-1},\vec{u}^{\anc{j}}_{0:t-1}} $ obeys the standard Kalman estimator update equations given in \eqref{pr:leaf2}. We can repeat the above arguments for all leaf nodes. This completes the proof of $P(0)$.

\subsection{Induction step: proof of $P(s)\implies P(s+1)$}\label{sec:Ps}

Suppose that $P(s)$ holds for some $s\ge 0$. We will prove $P(s+1)$ by sequentially considering each of the nodes in $\mathcal{G}^{s+1}$. Note that if $\mathcal{G}^{s+1}$ is empty then  $P(s)$ and $P(s+1)$ are equivalent and the induction step is trivially complete\footnote{In fact, if $\mathcal{G}^{s+1}$ is empty, then it is easy to show that all subsequent generations are empty as well and hence $\mathcal{G}^{\leq s} = \mathcal{V}$. This implies that Theorems~\ref{thm:main} and~\ref{thm:main_2} can be directly obtained from $P(s)$. }.  Therefore, we focus on the case $\mathcal{G}^{s+1} \neq \emptyset$.

We now focus on a particular node $k \in \mathcal{G}^{s+1}$ and its descendants. Note that $\vec{u}^k_t = f^k_t(\vec{i}^{\anc{k}}_t)$, which can be written as 
\begin{equation} \vec{u}^k_t = \sum_{b \in \anc{k}} h^{kb}(\vec{i}^b_t), \label{eq:precursor_0}
\end{equation}
for some linear functions $h^{kb}(\cdot)$.

 If $j \in \sdes{k}$, then, by definition of $\mathcal{G}^{s+1}$, we must have $j \in \mathcal{G}^{\leq s}$. Therefore, by our induction hypothesis, controller $j$'s strategy has the structure specified by \eqref{eq:induction_hypothesis2}. Decomposing the second summation in \eqref{eq:induction_hypothesis2}, we obtain
\begin{align}
  \vec u^j_t &= \sum_{a \in \anc{j}\cap\mathcal{G}^{\leq s}}g^{ja}_t(\vec{z}^{\funnel{a}}_t) +\sum_{b \in \sanc{j}\cap\mathcal{G}^{\geq s+1}} h_t^{jb}(\vec i^b_t) \notag\\
  &= \sum_{a \in \anc{j}\cap\mathcal{G}^{\leq s}}g^{ja}_t(\vec{z}^{\funnel{a}}_t) + \sum_{b \in \anc{k}} h_t^{jb}(\vec i^b_t)  +\sum_{b \in \sanc{j}\cap\mathcal{G}^{\geq s+1}\setminus\anc{k}} h_t^{jb}(\vec i^b_t) \notag\\
  &= \sum_{a \in \anc{j}\cap\mathcal{G}^{\leq s}}g^{ja}_t(\vec{z}^{\funnel{a}}_t) + \vec{\tilde{u}}_t^{jk} +\sum_{b \in \sanc{j}\cap\mathcal{G}^{\geq s+1}\setminus\anc{k}} h_t^{jb}(\vec i^b_t)
  \label{turtle}
\end{align}
where we defined
\begin{equation} \label{eq:utilde_def}
\vec{\tilde{u}}^{jk}_t \defeq \sum_{b \in \anc{k}} h_t^{jb}(\vec i^b_t).
\end{equation}
Note that $\vec{\tilde u}_t^{jk}$ is the part of the decision $\vec{u}_t^j$ that depends on information available to node~$k$. Using this observation, we will proceed as follows:
\begin{enumerate}
\item[(i)] For all $j \in \sdes{k}$, we fix the $g^{ja}$ and $h^{jb}$ functions in~\eqref{turtle} to arbitrary linear functions.
\item[(ii)] We will focus on the control problem of optimally selecting $\vec{u}^k_t, \{\vec{\tilde u}^{jk}_t\}_{j \in \sdes{k}}$ based on the information $\vec{i}^{\anc{k}}_t$, which is the information common among node $k$ and all its descendants. In other words, we want to optimize the functions $h^{kb}(\cdot), h^{jb}(\cdot)$ appearing in \eqref{eq:precursor_0} and  \eqref{eq:utilde_def} while keeping all other parts of the strategy profile fixed.
\item[(iii)] Since $\vec{u}^k_t, \{\vec{\tilde u}^{jk}_t\}_{j \in \sdes{k}}$ are all to be selected based on the same information $\vec{i}^{\anc{k}}_t$, we can view this problem as a coordinated system in which there is a single coordinator at node $k$ who knows $\vec{i}^{\anc{k}}_t$ and decides $\vec{u}^k_t, \{\vec{\tilde u}^{jk}_t\}_{j \in \sdes{k}}$.
\end{enumerate}
 
 We will make one more observation before analyzing the coordinator's problem. Note that for $i \in \sdes{j}$ and $j \in \sdes{k}$, $\vec{\hat u}_t^{ij}$ in~\eqref{eq:utilde_0} may be expressed in terms of $\vec{\tilde{u}}^{ik}_t$ as follows.
\begin{align}
\vec{\hat u}^{ij}_t
&= \sum_{a \in \anc{j}\cap\mathcal{G}^{\leq s}}g^{ia}_t(\vec{z}^{\funnel{a}}_t) + \sum_{b \in \anc{i} \cap \sdes{j}} g^{ib}_t(E^{b,j}\vec{z}^\funnel{j}_t)  +\sum_{c \in \sanc{j}\cap\mathcal{G}^{\geq s+1}} h^{ic}_t (\vec i^c_t) \label{bear} \\
&= \sum_{a \in \anc{j}\cap\mathcal{G}^{\leq s}}g^{ia}_t(\vec{z}^{\funnel{a}}_t) + \sum_{b \in \anc{i} \cap \sdes{j}} g^{ib}_t(E^{b,j}\vec{z}^\funnel{j}_t)  +\left[ \sum_{c \in \anc{k}}h^{ic}_t (\vec i^c_t)  + \sum_{c \in \sanc{j}\cap\mathcal{G}^{\geq s+1}\setminus{\anc{k}}} h^{ic}_t (\vec i^c_t)\right] \notag \\
&= \sum_{a \in \anc{j}\cap\mathcal{G}^{\leq s}}g^{ia}_t(\vec{z}^{\funnel{a}}_t) + \sum_{b \in \anc{i} \cap \sdes{j}} g^{ib}_t(E^{b,j}\vec{z}^\funnel{j}_t)  + \left[\vec{\tilde u}_t^{ik} + \sum_{c \in \sanc{j}\cap\mathcal{G}^{\geq s+1}\setminus{\anc{k}}} h^{ic}_t (\vec i^c_t)\right], \label{eq:hat_tilde}
\end{align}
where we used the definition of $\vec{\tilde u}_t^{ik}$ from \eqref{eq:utilde_def}. This is permitted because node $i$ too belongs to $\sdes{k}$.

As stated in the following lemma, the coordinator's problem may be viewed as an instance of the centralized Problem \ref{prob:CLQG_6node}.

\begin{lem}\label{lem:Gs_lemma}
In Problem~\ref{prob:NPLQG}, assume that all nodes in $\mathcal{G}^{\leq s}$ are using strategies of the form prescribed by Proposition $P(s)$. Pick any $k \in \mathcal{G}^{s+1}$. Fix the linear strategies for all nodes $a \notin \des{k}$. For $j\in \sdes{k}$, also fix the $g^{ja}$ and $h^{jb}$ in~\eqref{turtle}. Then the optimization problem for the coordinator at node $k$ is an instance of the six-node centralized problem with states, measurements, and inputs given by
\begin{align}\label{eq:newcoord_to_6_node}
\begin{aligned} 
\bar{\vec{x}}^1_t &=
	\bmat{\vec{x}^{\sanc{k}}_t \\ \vec{i}^{\sanc{k}}_t}  & 
\bar{\vec{x}}^2_t &=
	\bmat{\vec{x}^\coparent{k}_t \\ \vec{i}^\coparent{k}_t}  &
\bar{\vec{x}}^3_t &= \vec{x}^k_t  \\
\bar{\vec{x}}^4_t &=
	\bmat{\vec{x}^{\nonrelative{k}}_t \\ \vec{i}_t^{\nonrelative{k}}} &
\bar{\vec{x}}^5_t &=
	\bmat{\vec{x}^\sdes{k}_t\\\{\vec{z}^{\funnel{i}}_t\}_{i \in \sdes{k}}} &
\bar{\vec{x}}^6_t &= \bmat{\vec{x}^{\sibling{k}}_t \\ \vec{i}_t^{\sibling{k}}}
\end{aligned}
& & \bar{\vec{y}}_t^3 &= 
\bmat{\vec{y}_t^k \\ \vec{y}^\sanc{k}_t \\\vec{u}_t^\sanc{k}} & 
\bar{\vec{u}}_t^3 &= \bmat{ \vec{u}_t^k \\ \{\vec{\tilde u}_t^{ik}\}_{i \in \sdes{k}} }
\end{align}
\end{lem}

\begin{proof}
See Appendix \ref{sec:Gs_lemma}.
\end{proof}

Lemma~\ref{lem:Gs_lemma} establishes the coordinator's problem at node~$k$ to be a six-node centralized problem. This allows us to prove a precursor to  $P(s+1)$; a refinement of $P(s)$ that takes into account the node $k\in \mathcal{G}^{s+1}$.

\begin{lem}\label{lem:Ppre}
Suppose $P(s)$ holds and all nodes in $\mathcal{G}^{\leq s}$ are using strategies of the form prescribed by $P(s)$.  Consider a node $k\in\mathcal{G}^{s+1}$. Then
\begin{enumerate}
\item There is no loss in optimality in (further) restricting all $j\in \des{k}$ to strategies of the form
\begin{align*}
\vec u_t^j &= \sum_{a \in \anc{j}\cap(\mathcal{G}^{\leq s}\cup \{k\})}g^{ja}_t(\vec{z}^{\funnel{a}}_t) +\sum_{b \in \sanc{j}\cap(\mathcal{G}^{\geq s+1}\setminus\{k\})} h_t^{jb}(\vec i^b_t)
\end{align*}

\item For each $j \in \des{k}$, the dynamics of $\vec{z}^{\funnel{j}}_t = \eecc{\vec{x}^{\funnel{j}}_t}{\vec{i}^{\anc{j}}_t}$ are given by
\begin{equation}\label{eq:z_Gs_dynamics2}
\begin{aligned}
\vec{z}^{\funnel{j}}_0 &= 0 \\
\vec{z}^{\funnel{j}}_+ &\eqt
A^{\funnel{j}\funnel{j}}\vec{z}^{\funnel{j}}
 + B^{\funnel{j}\funnel{j}}\bmat{ \vec{u}^{\anc{j}} \\ \{\vec{\hat u}^{ij}\}_{i \in \sdes{j}}} - {L}^j\left(\vec{y}^{\anc{j}} -C^{\anc{j}\anc{j}}\vec{z}^{\anc{j}}\right)
\end{aligned}
\end{equation}
where for each $i\in\sdes{j}$, 
\begin{equation} 
 \vec{\hat u}^{ij}_t =  \sum_{a \in \anc{j}\cap(\mathcal{G}^{\leq s} \cup \{k\})}g^{ia}_t(\vec{z}^{\funnel{a}}_t) + \sum_{b \in \anc{i} \cap \sdes{j}} g^{ib}_t(E^{b,j}\vec{z}^\funnel{j}_t)  +\sum_{c \in \sanc{j}\cap(\mathcal{G}^{\geq s+1}\setminus {k})} h^{ic}_t (\vec i^c_t) \label{eq:newhateq}
\end{equation}
and the matrices ${L}^j_t$ in \eqref{eq:z_Gs_dynamics2} depend only on $\{g^{ib}: i\in \sdes{j}, b \in \anc{i} \cap \sdes{j}\}$.

\end{enumerate}
\end{lem}

\begin{proof}
By Lemma~\ref{lem:Gs_lemma}, the coordinator's problem at node~$k$ is a six-node centralized problem. We may therefore apply Theorem~\ref{thm:six_node_result}. It follows that $\vec{u}_t^k, \{\vec{\tilde u}^{jk}_t\}_{j \in \sdes{k}}$ 
are linear functions of
\begin{equation}
\eecs{\bar{\vec x}^1_t,\bar{\vec x}^3_t, \bar{\vec x}^5_t}{\vec{i}^{\anc{k}}_t}
=\eecs{\vec{x}^{\sanc{k}}_t, \vec{i}^{\sanc{k}}_t, \vec{x}^k_t, \vec{x}^\sdes{k}_t,\{\vec{z}^{\funnel{i}}_t\}_{i \in \sdes{k}} }{\vec{i}^{\anc{k}}_t}\label{eq:precursor1}
\end{equation}
Rearranging the random vectors in the conditional expectation on the right hand side of \eqref{eq:precursor1}, we can state that $\vec{u}_t^k, \{\vec{\tilde u}^{jk}_t\}_{j \in \sdes{k}}$ 
are linear functions of (i) $\eecs{\vec{x}^{\funnel{k}}_t}{\vec{i}^{\anc{k}}_t}$, (ii) $\eecs{\vec{i}^{\sanc{k}}_t}{\vec{i}^{\anc{k}}_t}$ and (iii) $\eecs{\{\vec{z}^{\funnel{i}}_t\}_{i \in \sdes{k}}}{\vec{i}^{\anc{k}}_t}$. We look at each of the three terms separately. The first term is, by definition, $\vec{z}^{\funnel{k}}_t$. The second term is $\vec{i}^{\sanc{k}}_t$ because $\vec{i}^{\sanc{k}}_t \subset \vec{i}^{\anc{k}}_t$. For the third term,
since $\vec{z}^{\funnel{i}}_t = \eecs{\vec x^{\funnel{i}}_t}{\vec i^{\anc{i}}_t}$ and $\vec i^{\anc{i}}_t \supset \vec i^{\anc{k}}_t$ for $i \in \sdes{k}$, it follows from the smoothing property of conditional expectations that
$\eecs{\vec{z}^{\funnel{i}}_t}{\vec{i}^{\anc{k}}_t}  =
\mathbb{E}\bigl[ \eecs{\vec x^{\funnel{i}}_t}{\vec i^{\anc{i}}_t} \,\big|\, \vec{i}^{\anc{k}}_t \bigr] =
\eecs{\vec x^{\funnel{i}}_t}{\vec i^{\anc{k}}_t}$.
Furthermore, if we partition the ancestors of $i$ as
$\anc{i} = \anc{k} \cup (\anc{i}\cap\sdes{k}) \cup (\sanc{i}\setminus \funnel{k})$, we have
\begin{align} 
\eecs{\vec x^{\funnel{i}}_t}{\vec i^{\anc{k}}_t}
&= \left(\eecs{\vec x^{\sdes{i}}_t}{\vec i^{\anc{k}}_t},
\eecs{\vec x^{\anc{k}}_t}{\vec i^{\anc{k}}_t},
\eecs{\vec x_t^{\anc{i}\cap\sdes{k}}}{\vec i^{\anc{k}}_t},
\eecs{\vec x_t^{\sanc{i}\setminus\funnel{k}}}{\vec i^{\anc{k}}_t} \right) \notag\\
&= \left( \eecs{\vec x^\sdes{i}_t}{\vec i^{\anc{k}}_t},
\eecs{\vec x^{\anc{k}}_t}{\vec i^{\anc{k}}_t},
\eecs{\vec x_t^{\anc{i}\cap\sdes{k}}}{\vec i^{\anc{k}}_t}, 
\vec{0}\right), \label{eq:precursor_2}
\end{align}
where we used the fact that $\eecs{\vec x_t^{\sanc{i}\setminus\funnel{k}}}{\vec i^{\anc{k}}_t} = \vec{0}$ because of the uncorrelated noise condition in Assumption A2. (Note that node $k$ and  a node $a \in \sanc{i}\setminus\funnel{k}$ have a common descendant namely node $i$ but cannot have a common ancestor by Assumption A1. The same conclusion could also be inferred from inspecting the aggregated graph centered at node $k$.) The remaining non-zero terms in \eqref{eq:precursor_2} are components from the vector $\vec{z}^{\funnel{k}}_t$.


 We  therefore conclude that $\vec{u}_t^k, \{\vec{\tilde u}^{jk}_t\}_{j \in \sdes{k}}$ are linear functions of $\vec{z}_t^{\funnel{k}}$ and $\vec{i}^{\sanc{k}}_t$. That is, the optimal form of \eqref{eq:precursor_0} and \eqref{eq:utilde_def} is
 \begin{align}
&\vec{u}^k_t = g^{kk}_t(\vec{z}^{\funnel{k}}_t) + \sum_{b \in \sanc{k}} h^{kb}(\vec{i}^b_t)
  \notag \\
 &\vec{\tilde u}^{jk}_t = g^{jk}_t(\vec{z}^{\funnel{k}}_t) + \sum_{b \in \sanc{k}} h^{jb}(\vec{i}^b_t), ~~~\mbox{for $j \in \sdes{k}$}, \label{eq:kcoord1}
 \end{align} 
 for some linear functions $g^{kk}_t(\cdot), g_t^{jk}(\cdot)$ and $\{h_t^{kb}(\cdot),h_t^{jb}(\cdot)\}_{b\in\sanc{k}}$.
 Therefore, we have from~\eqref{turtle} that for $j \in \sdes{k}$,
\begin{align}
\vec{u}_t^{j}
&=\notag \sum_{a \in \anc{j}\cap\mathcal{G}^{\leq s}}g^{ja}_t(\vec{z}^{\funnel{a}}_t) + \vec{\tilde{u}}_t^{jk} +\sum_{b \in \sanc{j}\cap\mathcal{G}^{\geq s+1}\setminus\anc{k}} h_t^{jb}(\vec i^b_t) \\
&=\notag \sum_{a \in \anc{j}\cap\mathcal{G}^{\leq s}}g^{ja}_t(\vec{z}^{\funnel{a}}_t) + \left[ g_t^{jk}(\vec z_t^{\funnel{k}}) +  \sum_{b\in\sanc{k}} h_t^{jb}(\vec i_t^{b})\right] + \sum_{b \in \sanc{j}\cap\mathcal{G}^{\geq s+1}\setminus\anc{k}} h_t^{jb}(\vec i^b_t) \\
&=\label{mouse} \sum_{a \in \anc{j}\cap(\mathcal{G}^{\leq s}\cup \{k\})}g^{ja}_t(\vec{z}^{\funnel{a}}_t) +\sum_{b \in \sanc{j}\cap\mathcal{G}^{\geq s+1}\setminus\{k\}} h_t^{jb}(\vec i^b_t)
\end{align}
 This completes the proof of the first part of Lemma~\ref{lem:Ppre}. 

We now prove the second part of Lemma~\ref{lem:Ppre}. We first consider $j \in \sdes{k}$.  By our induction hypothesis, the estimation dynamics for node $j$ are as given in Proposition $P(s)$ (see equation \eqref{eq:z_Gs_dynamics}). We now use  \eqref{eq:kcoord1}  in \eqref{eq:hat_tilde}  to get the expression for $\hat{\vec{u}}^{ij}_t$ as given in \eqref{eq:newhateq}.

 Our final task is to prove the estimator dynamics for node $k$. For a given $j\in\sdes{k}$, we rearrange~\eqref{turtle} to separate the terms that depend on $\funnel{k}$ from those that do not.
\begin{align}
\vec u_t^j &= \sum_{a \in \anc{j}\cap\mathcal{G}^{\leq s}}g^{ja}_t(\vec{z}^{\funnel{a}}_t) + \vec{\tilde{u}}_t^{jk} +\sum_{b \in \sanc{j}\cap\mathcal{G}^{\geq s+1}\setminus\anc{k}} h_t^{jb}(\vec i^b_t) \notag \\
&=\sum_{a \in \anc{j}\cap\mathcal{G}^{\leq s} \cap \sdes{k} } g^{ja}_t(\vec{z}^{\funnel{a}}_t) +  \sum_{a \in \sanc{j}\cap\mathcal{G}^{\leq s} \setminus \sdes{k} } g^{ja}_t(\vec{z}^{\funnel{a}}_t) +
\vec{\tilde{u}}_t^{jk} + \sum_{b \in \sanc{j}\cap\mathcal{G}^{\geq s+1}\setminus\anc{k}} h_t^{jb}(\vec i^b_t) \notag \\
&=\left( \sum_{a \in \anc{j}\cap \sdes{k} } g^{ja}_t(\vec{z}^{\funnel{a}}_t) + \vec{\tilde{u}}_t^{jk} \right)  + \left( \sum_{a \in \sanc{j}\cap\mathcal{G}^{\leq s} \setminus \sdes{k} } g^{ja}_t(\vec{z}^{\funnel{a}}_t) + \sum_{b \in \sanc{j}\cap\mathcal{G}^{\geq s+1}\setminus\anc{k}} h_t^{jb}(\vec i^b_t) \right) \label{eq:precursor_4}
\end{align}
where in the final step we made use of the fact that $\anc{j}\cap(\mathcal{G}^{\leq s} \cap \sdes{k} ) = \anc{j}\cap\sdes{k}$. This follows because $k\in\mathcal{G}^{s+1}$ and therefore $\sdes{k} \subseteq \mathcal{G}^{\le s}$.

 The dynamics of $\vec{x}^{\funnel{k}}$ can  be written as
\begin{align}
\vec{x}^{\funnel{k}}_+
&\eqt A^{\funnel{k}\funnel{k}}\vec{x}^{\funnel{k}} + A^{\funnel{k}\coparent{k}}\vec{x}^{\coparent{k}} + B^{\funnel{k}\funnel{k}}\vec{u}^{\funnel{k}} + B^{\funnel{k}\coparent{k}}\vec{u}^{\coparent{k}} + \vec{w}^{\funnel{k}} \notag \\
&\eqt A^{\funnel{k}\funnel{k}}\vec{x}^{\funnel{k}} + A^{\funnel{k}\coparent{k}}\vec{x}^{\coparent{k}} + B^{\funnel{k}\funnel{k}}\bmat{\vec{u}^{\anc{k}} \\ \{\vec{u}^j\}_{j\in \sdes{k}}} + B^{\funnel{k}\coparent{k}}\vec{u}^{\coparent{k}} + \vec{w}^{\funnel{k}} \label{eq:precursor_3}
\end{align}
Using the expression for $\vec{u}^j$ from \eqref{eq:precursor_4} in \eqref{eq:precursor_3}, we get
\begin{align}
\vec{x}^{\funnel{k}}_+&\eqt A^{\funnel{k}\funnel{k}}\vec{x}^{\funnel{k}} +B^{\funnel{k}\funnel{k}} \bmat{\vec{u}^{\anc{k}} \\ \{\sum_{a \in \anc{j}\cap \sdes{k} } g^{ja}(\vec{z}^{\funnel{a}}) + \vec{\tilde{u}}^{jk}  \}_{j \in \sdes{k}}} + \vec{w}^{\funnel{k}} + \vec{n} \label{eq:Gs_est}
\end{align}
where 
\begin{equation*}
\vec{n} \defeqt  A^{\funnel{k}\coparent{k}}\vec{x}^{\coparent{k}} +B^{\funnel{k}\funnel{k}} \bmat{\vec{0} \\ \{\sum_{a \in \sanc{j}\cap\mathcal{G}^{\leq s} \setminus \sdes{k} } g^{ja}(\vec{z}^{\funnel{a}}) + \sum_{b \in \sanc{j}\cap\mathcal{G}^{\geq s+1}\setminus\anc{k}} h^{jb}(\vec i^b)\}_{j \in \sdes{k}}} + B^{\funnel{k}\coparent{k}}\vec{u}^{\coparent{k}}
\end{equation*}
 The key property about $\vec{n}_t$ is that it  is independent of $\vec{i}^{\anc{k}}_t$. This fact follows because of Assumption~A2 and the construction of $\vec{n}_t$ which ensured that it is a function of only random variables associated with nodes in $\coparent{k}$. We now derive the update equation of node $k$'s estimate in the following steps:
 \begin{enumerate}
 \item At time $t$, the coordinator at node  $k$ knows $\vec z^\funnel{k}_t = \eecs{\vec{x}^{\funnel{k}}_t}{\vec{i}^{\anc{k}}_t}$, and the fact that  $\eecs{\vec{w}^{\funnel{k}}_t}{\vec{i}^{\anc{k}}_t}  = 0$. Further, for a strict descendant $j$ of node $k$  ($j \in \sdes{k}$), controller $k$'s estimate of its strict descendant's estimate can be written as 
 \begin{align}
 \eecs{\vec{z}^{\funnel{j}}_t}{\vec{i}^{\anc{k}}_t}&= \eecs{\vec{x}^{\funnel{j}}_t}{\vec{i}^{\anc{k}}_t} \notag \\
 &=  \eecs{\vec{x}^{\funnel{j}}_t-E^{j,k}\vec{x}^{\funnel{k}}_t}{\vec{i}^{\anc{k}}_t} +\eecs{E^{j,k}\vec{x}^{\funnel{k}}_t}{\vec{i}^{\anc{k}}_t} \label{eq:Gsstep1}
 \end{align}
 where we used the matrices $E^{j,k}$ of Definition~\ref{def:Gmatrix} in \eqref{eq:Gsstep1}. The term $\vec{x}^{\funnel{j}}_t-E^{j,k}\vec{x}^{\funnel{k}}_t$ in \eqref{eq:Gsstep1} is nonzero only for components of $\vec{x}^{\anc{j}}_t$ that correspond to nodes $a \in \funnel{j} \setminus \funnel{k}$. Note that since $j \in \sdes{k}$ and $a \in \funnel{j} \setminus \funnel{k}$, we have $a \in \coparent{k}$.  Therefore, by the uncorrelated noise conditions of Assumption~A2,  the first conditional expectation in \eqref{eq:Gsstep1} is $\vec 0$. Thus,
 \begin{align}
 \eecs{\vec{z}^{\funnel{j}}_t}{\vec{i}^{\anc{k}}_t} = \eecs{E^{j,k}\vec{x}^{\funnel{k}}_t}{\vec{i}^{\anc{k}}_t} = E^{j,k}\vec z^\funnel{k}_t
 \end{align}
 \item Let ${u}^k_t$ and ${\tilde u^{jk}}_t$ be realizations of the decisions made by the coordinator at node $k$.
 \item  The coordinator at node $k$ gets a new vector of observations at time $t$, $\vec y^\anc{k}_t$.  Because 
\begin{align}
   \vec{y}^{\anc{k}}_{t} = C^{\anc{k}\anc{k}}_t\vec{x}^{\anc{k}}_t + \vec{v}^{\anc{k}}_t, \label{eq:gs_est_update2}
  \end{align} 
  this observation vector is correlated with $\vec{x}^{\anc{k}}_t$, $\vec{w}^{\funnel{k}}_t$, and $\{\vec{z}^{\funnel{j}}_t\}_{j \in \sdes{k}}$. Using the fact that $\eecc{\vec{y}^{\anc{k}}_{t}}{\vec{i}^{\anc{k}}_t} =\eecc{C^{\anc{k}\anc{k}}_t\vec{x}^{\anc{k}}_t + \vec{v}^{\anc{k}}_t}{\vec{i}^{\anc{k}}_t} = C^{\anc{k}\anc{k}}_t\vec{z}^{\anc{k}}_t$, the  controller can use the new observation vector    to obtain the quantities $\zeta^x_t$, $\zeta_t^w$, $\{\zeta_t^{z^j}\}_{j\in\sdes{k}}$ defined below.
\begin{equation}\label{eq:gsint_est}
\begin{gathered}
\begin{aligned} 
\zeta^x_t &\defeq  \eecs{\vec x^\funnel{k}_t}{ \vec{i}^{\anc{k}}_t, \vec y^\anc{k}_{t}} \\
&\phantom{:}=   \eecs{\vec x^\funnel{k}_t}{ \vec{i}^{\anc{k}}_t}- \mathcal{L}^x_t(\vec y^\anc{k}_t -  C^{\anc{k}\anc{k}}_t \vec z^\anc{k}_t) \\
&\phantom{:}=\vec z^\funnel{k}_t - \mathcal{L}^x_t(\vec y^\anc{k}_t -  C^{\anc{k}\anc{k}}_t \vec z^\anc{k}_t)
\end{aligned}
\quad
\begin{aligned}
\zeta^w_t &\defeq  \eecs{\vec w^\funnel{k}_t}{ \vec{i}^{\anc{k}}_t, \vec y^\anc{k}_{t}} \\
&\phantom{:}=\eecs{\vec w^\funnel{k}_t}{ \vec{i}^{\anc{k}}_t} - \mathcal{L}^w_t(\vec y^\anc{k}_t -  C^{\anc{k}\anc{k}}_t \vec z^\anc{k}_t) \\
&\phantom{:}=\vec{0} - \mathcal{L}^w_t(\vec y^\anc{k}_t -  C^{\anc{k}\anc{k}}_t \vec z^\anc{k}_t)
\end{aligned}
\\[3mm]
\begin{aligned}
\zeta^{z^j}_t &\defeq \eecs{\vec z^\funnel{j}_t}{ \vec{i}^{\anc{k}}_t, \vec y^\anc{k}_{t}}\\
&\phantom{:}= \eecs{\vec z^\funnel{j}_t}{ \vec{i}^{\anc{k}}_t} - \mathcal{L}^{z^j}_t(\vec y^\anc{k}_t -  C^{\anc{k}\anc{k}}_t \vec z^\anc{k}_t)\\
&\phantom{:}= E^{j,k}\vec z^\funnel{k}_t - \mathcal{L}^{z^j}_t(\vec y^\anc{k}_t -  C^{\anc{k}\anc{k}}_t \vec z^\anc{k}_t) \qquad\text{for}\quad j \in \sdes{k}
\end{aligned}
\end{gathered}
\end{equation}
where the matrices $\mathcal{L}^x_t, \mathcal{L}^w_t,\mathcal{L}^{z^j}_t$ depend on the conditional covariance matrices of $\vec{y}^{\anc{k}}_{t}$ given $\vec{i}^{\anc{k}}_t$ and the conditional cross-covariance of  $(\vec x^\funnel{k}_t,\vec w^\funnel{k}_t, \vec z^\funnel{j}_t)$ and $\vec{y}^{\anc{k}}_t$ given $\vec{i}^{\anc{k}}_t$.  The conditional cross-covariance of $\vec w^\funnel{k}_t$ and $\vec{y}^{\anc{k}}_{t}$ can be easily shown to be equal to the unconditional covariance of $\vec{w}^\funnel{k}_t$ and $\vec{v}^{\anc{k}}_{t}$. Because of Lemma \ref{lem:estimation} part 5, the conditional covariance matrices of $\vec{y}^{\anc{k}}_{t}$ given $\vec{i}^{\anc{k}}_t$ and the conditional cross-covariance of  $(\vec x^\funnel{k}_t, \vec z^\funnel{j}_t)$ and $\vec{y}^{\anc{k}}_t$ given $\vec{i}^{\anc{k}}_t$ depend only on $\{g^{jb}: j \in \sdes{k}, b \in \anc{j} \cap \sdes{k}\}$ since all the other functions  in descendants' strategies either  use only $\vec{i}^{\anc{k}}_t$ or only $\vec{i}^{\coparent{k}}_t$.
This step is the basic conditional mean update equation of a Gaussian random vector.
\item The state $\vec{x}^{\funnel{k}}$ evolves according to \eqref{eq:Gs_est}. Therefore,
\begin{align}
\vec{z}^{\funnel{k}}_{t+1} 
&= \eecs{\vec{x}^{\funnel{k}}_{t+1}}{\vec{i}^{\anc{k}}_{t},\vec{y}^{\anc{k}}_t} \notag \\
&= A^{\funnel{k}\funnel{k}}\zeta^x_t +B^{\funnel{k}\funnel{k}}\bmat{{u}^{\anc{k}} \\
 \{\sum_{a \in \anc{j}\cap\sdes{k}}g^{ja}_t(\zeta^{z^a}_t)+ {\tilde u}_t^{jk}\}_{j \in \sdes{k}}} + \zeta^w_t \label{eq:gsint_est2},
\end{align}
Using \eqref{eq:gsint_est} in \eqref{eq:gsint_est2} and rearranging the terms that involve $\left(\vec{y}^{\anc{k}}_t -C^{\anc{k}\anc{k}}\vec{z}^{\anc{k}}\right) $, we obtain the estimator dynamics
\begin{align}\label{eq:gsint_est3}
\vec{z}^{\funnel{k}}_+ &\eqt
A^{\funnel{k}\funnel{k}}\vec{z}^{\funnel{k}}
 + B^{\funnel{k}\funnel{k}}\bmat{ u^{\anc{k}} \\ \{\sum_{a \in \anc{j}\cap\sdes{k}}g^{ja}(E^{a,k}\vec{z}^{\funnel{k}})+{\tilde u}^{jk}\}_{j \in \sdes{k}}} - {L}^k\left(\vec{y}^{\anc{k}} -C^{\anc{k}\anc{k}}\vec{z}^{\anc{k}}\right)
\end{align}
where the matrices $L^k_t$ in \eqref{eq:gsint_est3}  depend only on $\{g^{jb}: j \in \sdes{k}, b \in \anc{j} \cap \sdes{k}\}$.
\end{enumerate}
Substituting $\vec{\tilde u}^{jk} = g_t^{jk}(\vec z_t^{\funnel{k}}) +  \sum_{b\in\sanc{k}} h_t^{jb}(\vec i_t^{b})$, we obtain that the estimator dynamics for node $k$ conform to \eqref{eq:z_Gs_dynamics2}. This establishes the required estimator dynamics for node $k$.
\end{proof}

We will now show how to recursively apply Lemma~\ref{lem:Ppre} to prove $P(s+1)$. Suppose that $\mathcal{G}^{s+1} = \{k_1,\dots,k_\ell\}$.  We start with using the result of Lemma~\ref{lem:Ppre} for $k_1$. Now consider $k_2$. If $\sdes{k_2} \cap \sdes{k_1} = \emptyset$, we can repeat the argument of Lemma~\ref{lem:Ppre} to get an analogous result.  Consider then the case when there exists some $j \in \sdes{k_2} \cap \sdes{k_1} $.
  Lemma~\ref{lem:Ppre} implies that the control strategy for $j$ can be written as
\begin{align*}
\vec u_t^j &= \sum_{a \in \anc{j}\cap(\mathcal{G}^{\leq s}\cup \{ k_1\})}g^{ja}_t(\vec{z}^{\funnel{a}}_t) +\sum_{b \in \sanc{j}\cap(\mathcal{G}^{\geq s+1}\setminus\{k_1\})} h_t^{jb}(\vec i^b_t) \\
&= \sum_{a \in \anc{j}\cap(\mathcal{G}^{\leq s}\cup \{k_1\})}g^{ja}_t(\vec{z}^{\funnel{a}}_t)
+ \vec{\tilde u}_t^{j{k_2}}
+\sum_{b \in \sanc{j}\cap(\mathcal{G}^{\geq s+1}\setminus\{k_1,\anc{k_2}\})} h_t^{jb}(\vec i^b_t)
\end{align*}
where we defined
\[
\vec{\tilde{u}}^{jk_2}_t \defeq \sum_{b \in \anc{k_2}} h_t^{jb}(\vec i^b_t)
\]
Similarly, $\vec{\hat{u}}^{ij}_t$ of \eqref{eq:newhateq} can be written as
\begin{align*}
\vec{\hat u}^{ij}_t =  \sum_{a \in \anc{j}\cap(\mathcal{G}^{\leq s} \cup \{k_1\})}g^{ia}_t(\vec{z}^{\funnel{a}}_t) + \sum_{b \in \anc{i} \cap \sdes{j}} g^{ib}_t(E^{b,j}\vec{z}^\funnel{j}_t)  +\vec{\tilde u}_t^{j{k_2}}
+\sum_{b \in \sanc{j}\cap(\mathcal{G}^{\geq s+1}\setminus\{k_1,\anc{k_2}\})} h_t^{jb}(\vec i^b_t)
\end{align*}
We can now consider a coordinator at node $k_2$ that selects $\vec{u}^{k_2}_t, \{\vec{\tilde u}^{ik_2}_t\}_{i \in \sdes{k_2}}$ and use the same argument as in Lemma~\ref{lem:Ppre} to argue that $\vec{u}^{k_2}_t, \{\vec{\tilde u}^{ik_2}_t\}_{i \in \sdes{k_2}}$ can be optimally chosen as functions of $\vec z^\funnel{k_2}_t, \vec i^\sanc{k_2}_t$ and that the estimator dynamics for $\vec z^\funnel{k_2}_t$ are analogous to the estimator dynamics given in Lemma~\ref{lem:Ppre} . Note that because the estimator for the coordinator at node $k_1$  depended only on $\{g^{pq}: p,q \in \sdes{k_1}\}$, varying the strategy for coordinator at node $k_2$ will have no effect on  estimation at node $k_1$. 
The result is that $\vec u_t^j$ is of the form
\begin{align*}
\vec u_t^j &= \sum_{a \in \anc{j}\cap(\mathcal{G}^{\leq s}\cup \{k_1,k_2\})}g^{ja}_t(\vec{z}^{\funnel{a}}_t) +\sum_{b \in \sanc{j}\cap(\mathcal{G}^{\geq s+1}\setminus\{k_1,k_2\})} h_t^{jb}(\vec i^b_t) 
\end{align*}
and $\vec{\hat u}^{ij}_t$ in the estimation dynamics for $j \in \mathcal{G}^{\leq s}$ are now
\begin{align*}
\vec{\hat u}^{ij}_t =  \sum_{a \in \anc{j}\cap(\mathcal{G}^{\leq s} \cup \{k_1,k_2\})}g^{ia}_t(\vec{z}^{\funnel{a}}_t) + \sum_{b \in \anc{i} \cap \sdes{j}} g^{ib}_t(E^{b,j}\vec{z}^\funnel{j}_t)  +
\sum_{b \in \sanc{j}\cap(\mathcal{G}^{\geq s+1}\setminus\{k_1,k_2\})} h_t^{jb}(\vec i^b_t)
\end{align*}
It is now clear that the above argument can be repeated for $\{k_3,\dots,k_\ell\}$. We then obtain
\begin{align*}
\vec u_t^j &= \sum_{a \in \anc{j}\cap \mathcal{G}^{\leq s+1}}g^{ja}_t(\vec{z}^{\funnel{a}}_t) +\sum_{b \in \sanc{j}\cap \mathcal{G}^{\geq s+2}} h_t^{jb}(\vec i^b_t) 
\end{align*}
for $j \in \mathcal{G}^{\leq s+1}$ and $\vec{\hat u}^{ij}_t$ in the estimation dynamics for $j \in \mathcal{G}^{\leq s+1}$ as
\begin{align*}
\vec{\hat u}^{ij}_t =  \sum_{a \in \anc{j}\cap\mathcal{G}^{\leq s+1}}g^{ia}_t(\vec{z}^{\funnel{a}}_t) + \sum_{b \in \anc{i} \cap \sdes{j}} g^{ib}_t(E^{b,j}\vec{z}^\funnel{j}_t)  +
\sum_{b \in \sanc{j}\cap\mathcal{G}^{\geq s+2}} h_t^{jb}(\vec i^b_t)
\end{align*}
and this establishes $P(s+1)$.

\section{Proof modification under Assumption A2'}\label{sec:proof_alt}
Assumption A2 required that for any pair of nodes that neither have a common ancestor nor a common descendant must have decoupled costs and uncorrelated noise (see Definitions~\ref{def:indep_cost} and \ref{def:indep_noise}). In other words, for a given node $j$, each node $k \in \nonrelative{j}$ has both decoupled costs and uncorrelated noise with respect to node $j$. The relaxed form of this assumption---Assumption A2'---requires that each node $k \in \nonrelative{j}$, has \emph{either} decoupled costs \emph{or} uncorrelated noise. 

In order to modify the proof of our main results under assumption A2',  
we will modify the definition of aggregated graph centered at node $j$ (see Section~\ref{sec:aggregated}). We first partition the nodes in $\nonrelative{j}$ into 3 subsets:
\begin{enumerate}
\item $\nonrelative{j_a}$ is the set of all nodes $k \in \nonrelative{j}$ that have both decoupled costs and uncorrelated noise when compared with node $j$.
\item $\nonrelative{j_b}$ is the set of all nodes $k \in \nonrelative{j} \setminus \nonrelative{j}_a$ that have uncorrelated noise but not decoupled costs with respect to node $j$.
\item $\nonrelative{j_c}$ is the set of all nodes $k \in \nonrelative{j} \setminus \nonrelative{j}_a$ that have decoupled costs but not uncorrelated noise with respect to node $j$.
\end{enumerate}
We can now define a modified aggregated graph centered at node $j$ as follows:
\begin{defn}[modified six-node aggregated graph]
Suppose $G(\mathcal{V},\mathcal{E})$ is a DAG and Assumption~A1 and A2' hold. For  every $j \in \mathcal{V}$, we define the \eemph{modified six-node aggregated graph} centered at~$j$, denoted $G^j(\mathcal{V}^j,\mathcal{E}^j)$, as the multitree of Figure~\ref{fig:rel_aggregate}.
\begin{figure}[ht]
\centering
\begin{tikzpicture}[thick,>=latex]
	\tikzstyle{block}=[circle,draw,fill=black!6,minimum height=2.4em];
	\tikzstyle{cblock}=[cloud,draw,fill=black!6,cloud puffs=7.5,cloud puff arc=120, aspect=1.7, inner sep=1mm];
	
	\def\x{2.7};
	\def\y{1.1};
	\node [cblock](N1) at (-0.5*\x,1.2*\y) {$\textbf{1}:\sanc{j}$};
	\node [cblock, inner sep=-1.2mm](N2) at (0.5*\x,1.2*\y) {$\textbf{2}:\!\coparent{j} \cup \nonrelative{j_b}$};
	\node [block](N3) at (-1.2*\x,0) {$\textbf{3}:j$};
	\node [cblock](N4) at (1.2*\x,0) {$\textbf{4}:\nonrelative{j_a}$};
	\node [cblock](N5) at (-0.5*\x,-1.2*\y) {$\textbf{5}:\sdes{j}$};
    \node [cblock, inner sep=-1.2mm](N6) at (0.5*\x,-1.2*\y) {$\textbf{6}:\!\sibling{j}\cup \nonrelative{j_c}$};
    
	\draw [->] (N1) -- (N3);
	\draw [->] (N1) -- (N6);
	\draw [->] (N2) -- (N5);
    \draw [->] (N3) -- (N5);
	\draw [<->] (N2) -- (N6);
	\draw [<->] (N2) -- (N4);
	\draw [<->] (N4) -- (N6);
	\draw [->] (N1) -- (N5);
\end{tikzpicture}
\caption{Diagram of $G^j(\mathcal{V}^j,\mathcal{E}^j)$, the modified six-node aggregate centered at~$j$.\label{fig:rel_aggregate}}
\end{figure}
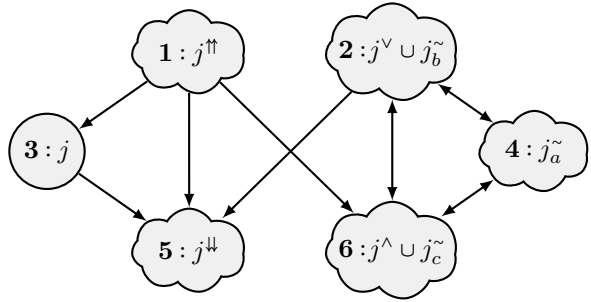
\end{defn}
Note that all nodes in the node labeled $\textbf{2}$ in the modified aggregated graph have uncorrelated noise with respect to node $j$, while all nodes in the node labeled $\textbf{6}$ in the modified aggregated graph have decoupled cost with respect to node $j$. Further, all nodes in the node labeled $\textbf{4}$ have both decoupled cost and uncorrelated noise with respect to node $j$. This property of the modified aggregated graph is the same as in the original aggregated graph of Section~\ref{sec:aggregated}.

The aggregated graph centered at node $j$ will serve as the basis for aggregating states and dynamics of the overall system when we want to focus on node $j$. 
We define aggregated states $\set{\bar{\mathbf{x}}^m}{m=1,\dots,6}$ corresponding to each node in $G^j(\mathcal{V}^j,\mathcal{E}^j)$ as follows:
\begin{align*}
\bar{\mathbf{x}}^1 &= \mathbf{x}^\sanc{j} &
\bar{\mathbf{x}}^2 &= \bmat{\mathbf{x}^\coparent{j}\\\mathbf{x}^\nonrelative{j_b}} &
\bar{\mathbf{x}}^3 &= \mathbf{x}^j &
\bar{\mathbf{x}}^4 &= \mathbf{x}^{\nonrelative{j_a}} &
\bar{\mathbf{x}}^5 &= \mathbf{x}^\sdes{j} &
\bar{\mathbf{x}}^6 &= \bmat{\mathbf{x}^\sibling{j} \\\mathbf{x}^\nonrelative{j_c}}
\end{align*}
and define $\bar{\mathbf{u}}$, $\bar{\mathbf{y}}$, $\bar{\mathbf{w}}$, and $\bar{\mathbf{v}}$ in a similar fashion. Once this is done, the aggregated dynamics centered at $j$ may be written compactly as
\begin{equation}\label{eq:state_eqns_aggregated_2}
\begin{aligned}
\bar{\mathbf{x}}_+ &\eqt \bar A\bar{\mathbf{x}} + \bar B\bar{\mathbf{u}} + \bar{\mathbf{w}} \\
\bar{\mathbf{y}} &\eqt \bar C\bar{\mathbf{x}} + \bar{\mathbf{v}}
\end{aligned}
\end{equation}
If we split the matrices $\bar A, \bar B$ and $\bar C$ into blocks according to the states of the modified aggregated graph, they will have a sparsity pattern corresponding to the graph of Figure~\ref{fig:rel_aggregate}.

With this redefinition of the aggregated graph, we can proceed along the same steps as in Sections~\ref{sec:proof_prelim} and~\ref{sec:proof_main} to establish Theorems \ref{thm:main} and \ref{thm:main_2} under Assumption A2'.

\section{Concluding remarks}\label{sec:conclusion}

In this first part of the paper, we described a broad class of decentralized output feedback LQG control problems that admit simple and intuitive sufficient statistics. Each controller must estimate the state of its ancestors and descendants in the underlying DAG. The optimal control action for each controller is a linear function of the estimate it computes as well as the estimates computed by all of its ancestors. Moreover, we proved that estimates can be computed recursively much like a Kalman filter.

 Several aspects of the problem architecture influence the structure of the optimal control strategies. These are:
\begin{enumerate}[F1.]
\itemsep=1mm
\item measurement type: output feedback or state feedback?
\item DAG topology: is it a multitree or not?
\item noise structure: which correlations between subsystems are permitted?
\item cost structure: which cost couplings between subsystems are permitted?
\end{enumerate}
These items delineate a subset of PN and QI problems, as in Figure~\ref{fig:venn}. In the present work, we assumed output feedback for all subsystems (item F1), but we imposed restrictions on the DAG topology, noise, and cost via Assumptions A1 and A2 or A2' (items F2--F4).

The prior works discussed in Section~\ref{sec:intro} can be classified based on which features F1--F4 are strengthened and which are relaxed. For example, the poset-causal framework of~\cite{shah_parrilo} considers a fully general DAG topology and quadratic cost (items F2 and F4), but this comes at the expense of requiring state-feedback for all nodes and uncorrelated noise among subsystems (items F1 and F3).

The complete understanding of which assumptions on features F1--F4 lead to simple optimal controllers remains an open issue. Indeed, some seemingly restrictive assumptions may still lead to controllers with complicated structures. For example, consider two fully decoupled subsystems, each with state feedback (items F1 and F2). Then assume the process noise is correlated between subsystems and the cost is also coupled (items F3 and F4). This problem was studied in~\cite{lessard_diagsf,lessard_diagsf2}, where it was shown that the optimal controller may have state dimension up to $n^2$, where $n$ is the global state dimension of the plant.

Part II of the paper takes the results of the present paper one step further and derives an explicit and efficiently computable state-space representation for the optimal controller. As with centralized LQG control problems, the optimal estimation and control gains may be be computable offline, and the computational complexity is similar as well.

\appendix
\section{Proof of Theorem \ref{thm:six_node_result}} \label{sec:proof_thm3}
Because of Assumption A2 in Problem~\ref{prob:CLQG_6node}, the cost  matrices $\{Q_{0:T-1},R_{0:T-1},S_{0:T-1},P_\textup{final}\}$ have blocks that conform to the sparsity pattern of 
 \[\bar\Sbin^\tp \bar\Sbin = \bmat{ 1 & 1 & 1 & 1 & 1 & 1 \\
			1 & 1 & 1 & 1 & 1 & 1 \\
			1 & 1 & 1 & 0 & 1 & 0 \\
			1 & 1 & 0 & 1 & 0 & 1 \\
			1 & 1 & 1 & 0 & 1 & 0 \\
			1 & 1 & 0 & 1 & 0 & 1 } .\] 
In other words, the node pairs $(3,4), (3,6),(4,5),(5,6)$ have \emph{decoupled cost} (see Definition~\ref{def:indep_cost}).			
			Therefore, the objective function $\hat{\mathcal{J}}_0(f^3)$ can be decomposed as 
 \begin{equation}
  \hat{\mathcal{J}}_0(f^3) = c + \tilde{\mathcal{J}}_0(f^3),
 \end{equation}
 where 
 \begin{align}
 c = \ee \biggl( \sum_t
 \bmat{\mathbf{x}^{1,2,4,6}}^\tp \tilde{Q} \bmat{\mathbf{x}^{1,2,4,6}}
 +  \bmat{\mathbf{x}^{1,2,4,6}}^\tp \tilde{P}_{\text{final}}  \bmat{\mathbf{x}^{1,2,4,6}} \biggr)
 \end{align}
 does not depend on the choice of control strategy and 
 \begin{align}
  \tilde{\mathcal{J}}_0(f^3) = \ee^{f^3} \biggl( \sum_t
  \bmat{\mathbf{x}^{1,2,3,5}\\\mathbf{u}^3}^\tp \bmat{\hat{Q} & \hat{S} \\ \hat{S}^\tp & \hat{R}} \bmat{\mathbf{x}^{1,2,3,5}\\ \mathbf{u}^3}
  + \bmat{\mathbf{x}^{1,2,3,5}}^\tp \hat{P}_\textup{final} \bmat{\mathbf{x}^{1,2,3,5}} \biggr)
  \end{align}
  Therefore, minimizing the objective function $\hat{\mathcal{J}}_0(f^3)$ is the same as minimizing $\tilde{\mathcal{J}}_0(f^3)$. Observing that both $\tilde{\mathcal{J}}_0(f^3)$ and the dynamics of $\VEC x^{1,2,3,5}$ do not involve $\VEC x^{4,6}$, we conclude that minimizing $\tilde{\mathcal{J}}_0(f^3)$ is a standard centralized LQG problem with $\VEC x^{1,2,3,5}$ as the state. Therefore, the optimal control strategy is a linear function of $\eecs{\VEC{x}^{1,2,3,5}}{\VEC{i}^3_t}$. Further, observe that the dynamics of $\VEC x^2$ are
  \begin{equation}
  \VEC{x}^2_+ \eqt  A_{22}\VEC{x}^2 +  \VEC{w}^2.
  \end{equation}
 Because of Assumption A2, the covariance matrices $\{W_{0:T-1},V_{0:T-1},U_{0:T-1},\Sigma_\text{init}\}$ have blocks that conform to the sparsity pattern of $\bar\Sbin\bar\Sbin^\tp$. In other words,  nodes $2$ and $3$ have \emph{uncorrelated noise} (see Definition \ref{def:indep_noise}) and hence, $\eecs{\VEC{x}^{2}}{\VEC{i}^3_t} = \ee(\VEC{x}^2)=0$. Therefore, the optimal control strategy is in fact a linear function of $\eecs{\VEC{x}^{1,3,5}}{\VEC{i}^3_t}$.    \eqref{eq:est_centralized} is essentially the centralized Kalman estimate update equation  for $\eecs{\VEC{x}^{1,3,5}}{\VEC{i}^3_t}$.

\section{Proof of Lemma \ref{lem:estimation}} \label{sec:proof_estimation}
\begin{enumerate}
\item  Part 1 follows from the sparsity assumptions about the covariance matrices which imply that under any linear strategy $\vec{x}^{\coparent{j}}_t,\vec{i}^{\coparent{j}}_t$ (and $\vec{x}^{\nonrelative{j}}_t,\vec{i}^{\nonrelative{j}}_t$) are independent of $\vec i^{\anc{j}}_t$.

\item With just $f^{\funnel{j}}, f^{\coparent{j}}$ fixed and $f^{\sibling{j}}, f^{\nonrelative{j}}$ left unspecified,  it is easy to check that all the random variables $\vec{x}^{\funnel{j}}_t,\vec{i}^{\funnel{j}}_t,\vec i^{\anc{j}}_t$  (and consequently their joint probability distribution) are well-defined. Therefore,  the conditional expectation $\eec{\vec{x}^{\funnel{j}}_t,\vec{i}^{\funnel{j}}_t}{\vec i^{\anc{j}}_t}$ cannot depend on the choice of $f^{\sibling{j}}, f^{\nonrelative{j}}$.

We will now write $\vec{x}^{\funnel{j}}_t,\vec{i}^{\funnel{j}}_t$ in terms of the primitive random variables and the control decisions:
\begin{align} 
(\vec{x}^{\funnel{j}}_t,\vec{i}^{\funnel{j}}_t) &\in \lin(\vec{x}^{\funnel{j}}_0,\vec{w}^{\funnel{j}}_{1:t-1},\vec{v}^{\funnel{j}}_{1:t-1}, \vec{u}^{\funnel{j}}_{1:t-1},\vec{x}^{\coparent{j}}_0,\vec{w}^{\coparent{j}}_{1:t-1},\vec{v}^{\coparent{j}}_{1:t-1}, \vec{u}^{\coparent{j}}_{1:t-1}) \notag \\
&\in \lin(\vec{x}^{\funnel{j}}_0,\vec{w}^{\funnel{j}}_{1:t-1},\vec{v}^{\funnel{j}}_{1:t-1}, \vec{u}^{\anc{j}}_{1:t-1},\vec{x}^{\coparent{j}}_0,\vec{w}^{\coparent{j}}_{1:t-1},\vec{v}^{\coparent{j}}_{1:t-1}, \vec{u}^{\coparent{j}}_{1:t-1})  + \lin(\vec{u}^{\sdes{j}}_{1:t-1} )\label{eq:est_graph0}
\end{align}
where the linear functions $\lin$ depend only on system parameters. Given the strategy $f^{\sdes{j}}$, we can write 
\begin{align} \label{eq:est_graph0a}
\vec{u}^{\sdes{j}}_{1:t-1}  \in \lin^{f^\sdes{j}}(\vec{x}^{\funnel{j}}_0,\vec{w}^{\funnel{j}}_{1:t-1},\vec{v}^{\funnel{j}}_{1:t-1}, \vec{u}^{\anc{j}}_{1:t-1},\vec{x}^{\coparent{j}}_0,\vec{w}^{\coparent{j}}_{1:t-1},\vec{v}^{\coparent{j}}_{1:t-1}, \vec{u}^{\coparent{j}}_{1:t-1})
\end{align}
where the linear function $\lin^{f^\sdes{j}}$ depends on ${f^\sdes{j}}$. Combining \eqref{eq:est_graph0} and \eqref{eq:est_graph0a} and regrouping terms, we can write
\begin{align}\label{eq:est_graph1}
 (\vec{x}^{\funnel{j}}_t,\vec{i}^{\funnel{j}}_t )\in \lin^{f^\sdes{j}}(\vec{x}^{\funnel{j}}_0,\vec{w}^{\funnel{j}}_{1:t-1},\vec{v}^{\funnel{j}}_{1:t-1}, \vec{u}^{\anc{j}}_{1:t-1} )+ \lin^{f^\sdes{j}}(\vec{x}^{\coparent{j}}_0,\vec{w}^{\coparent{j}}_{1:t-1},\vec{v}^{\coparent{j}}_{1:t-1}, \vec{u}^{\coparent{j}}_{1:t-1}) 
\end{align}
 
Because of the sparsity assumptions about covariance matrices, the second term in \eqref{eq:est_graph1} is independent of $\VEC i^{\anc{j}}_t$. Therefore, under any linear strategy $f$, the conditional expectation of the second term in \eqref{eq:est_graph1} given $\vec i^{\anc{j}}_t$ is $0$. 

By linearity of conditional expectation, 
\begin{align}
&\eec{\vec{x}^{\funnel{j}}_t,\vec{i}^{\funnel{j}}_t }{i^{\anc{j}}_t} \in \lin^{f^\sdes{j}}\left(\eec{\vec{x}^{\funnel{j}}_0,\vec{w}^{\funnel{j}}_{0:t-1},\vec{v}^{\funnel{j}}_{0:t-1}, \vec{u}^{\anc{j}}_{0:t-1}}{ i^{\anc{j}}_t}\right) \notag \\
& \in \lin^{f^\sdes{j}}\left(\eec{\vec{x}^{\funnel{j}}_0,\vec{w}^{\funnel{j}}_{0:t-1},\vec{v}^{\funnel{j}}_{0:t-1}, \vec{u}^{\anc{j}}_{0:t-1}}{ y^{\anc{j}}_{0:t-1},u^{\anc{j}}_{0:t-1}}\right) \notag \\
& \in \lin^{f^\sdes{j}}\left({u}^{\anc{j}}_{0:t-1}, \eec{\vec{x}^{\funnel{j}}_0,\vec{w}^{\funnel{j}}_{0:t-1},\vec{v}^{\funnel{j}}_{0:t-1}}{y^{\anc{j}}_{0:t-1},u^{\anc{j}}_{0:t-1}}\right) \label{eq:est_graph1c}
\end{align}

Note that $\vec{y}^{\anc{j}}_{0:t-1}$ can be written as $\mathcal{F}\bmat{\vec{x}^{\anc{j}}_0 \\\vec{w}^{\anc{j}}_{0:t-1} \\\vec{v}^{\anc{j}}_{0:t-1} } + \mathcal{G}\vec{u}^{\anc{j}}_{0:t-1}$, where $\mathcal{F}, \mathcal{G}$ depend only on system parameters and not on any node's strategy. Therefore, 
\begin{align}
& \eec{\vec{x}^{\funnel{j}}_0,\vec{w}^{\funnel{j}}_{0:t-1},\vec{v}^{\funnel{j}}_{0:t-1}}{y^{\anc{j}}_{0:t-1},u^{\anc{j}}_{0:t-1}} \notag \\
 &= \eec{\vec{x}^{\funnel{j}}_0,\vec{w}^{\funnel{j}}_{0:t-1},\vec{v}^{\funnel{j}}_{0:t-1}}{\mathcal{F}\bmat{\vec{x}^{\anc{j}}_0 \\\vec{w}^{\anc{j}}_{0:t-1} \\\vec{v}^{\anc{j}}_{0:t-1} } + \mathcal{G}\vec{u}^{\anc{j}}_{0:t-1} = y^{\anc{j}}_{0:t-1}, \vec{u}^{\anc{j}}_{0:t-1}=u^{\anc{j}}_{0:t-1}} \notag \\
 &= \eec{\vec{x}^{\funnel{j}}_0,\vec{w}^{\funnel{j}}_{0:t-1},\vec{v}^{\funnel{j}}_{0:t-1}}{\mathcal{F}\bmat{\vec{x}^{\anc{j}}_0 \\\vec{w}^{\anc{j}}_{0:t-1} \\\vec{v}^{\anc{j}}_{0:t-1} } + \mathcal{G}{u}^{\anc{j}}_{0:t-1} = y^{\anc{j}}_{0:t-1}} \label{eq:est_graph1b} \\
  &= \eec{\vec{x}^{\funnel{j}}_0,\vec{w}^{\funnel{j}}_{0:t-1},\vec{v}^{\funnel{j}}_{0:t-1}}{\mathcal{F}\bmat{\vec{x}^{\anc{j}}_0 \\\vec{w}^{\anc{j}}_{0:t-1} \\\vec{v}^{\anc{j}}_{0:t-1} }  = y^{\anc{j}}_{0:t-1}- \mathcal{G}{u}^{\anc{j}}_{0:t-1}} \label{eq:est_graph2}
  \end{align}
  where we removed $\vec{u}^{\anc{j}}_{0:t-1}$ from conditioning in \eqref{eq:est_graph1b} since it is a function of $\vec{y}^{\anc{j}}_{0:t-1}$ which is already included in the conditioning.
  Since all random variables appearing in \eqref{eq:est_graph2} are only primitive random variables, the conditional expectation does not depend on the choice of strategy $f$. Recall that the linear function in \eqref{eq:est_graph1c} depends only on $f^{\sdes{j}}$. Therefore, if $g^{\sdes{j}} = f^{\sdes{j}}$, then 
  \[ \ee^{{g}}\left(\vec{x}^{\funnel{j}}_t,\vec{i}^{\funnel{j}}_t\middle\vert\vec i^{\anc{j}}_t = i^{\anc{j}}_t\right)
  = \ee^{{f}}\left(\vec{x}^{\funnel{j}}_t,\vec{i}^{\funnel{j}}_t\middle\vert\vec i^{\anc{j}}_t = i^{\anc{j}}_t\right)\]
  
\item  Under the new strategy $ {g}$, we can write \eqref{eq:est_graph0a} as
 \begin{multline} \label{eq:est_graph3}
\vec{u}^{\sdes{j}}_{1:t-1}  \in \lin^{f^\sdes{j}}(\vec{x}^{\funnel{j}}_0,\vec{w}^{\funnel{j}}_{1:t-1},\vec{v}^{\funnel{j}}_{1:t-1}, \vec{u}^{\anc{j}}_{1:t-1},\vec{x}^{\coparent{j}}_0,\vec{w}^{\coparent{j}}_{1:t-1},\vec{v}^{\coparent{j}}_{1:t-1}, \vec{u}^{\coparent{j}}_{1:t-1}) 
\\+ \lin^{{m}^{\sdes{j}}}(\vec{x}^{\coparent{j}}_0,\vec{w}^{\coparent{j}}_{1:t-1},\vec{v}^{\coparent{j}}_{1:t-1}, \vec{u}^{\coparent{j}}_{1:t-1}) 
\end{multline}
and combine this with \eqref{eq:est_graph0} to get
\begin{multline}\label{eq:est_graph4}
 (\vec{x}^{\funnel{j}}_t,\vec{i}^{\funnel{j}}_t )\in \lin^{f^\sdes{j}}(\vec{x}^{\funnel{j}}_0,\vec{w}^{\funnel{j}}_{1:t-1},\vec{v}^{\funnel{j}}_{1:t-1}, \vec{u}^{\anc{j}}_{1:t-1} )+ \lin^{f^\sdes{j}}(\vec{x}^{\coparent{j}}_0,\vec{w}^{\coparent{j}}_{1:t-1},\vec{v}^{\coparent{j}}_{1:t-1}, \vec{u}^{\coparent{j}}_{1:t-1}) 
 \\+ \lin^{{m}^{\sdes{j}}}(\vec{x}^{\coparent{j}}_0,\vec{w}^{\coparent{j}}_{1:t-1},\vec{v}^{\coparent{j}}_{1:t-1}, \vec{u}^{\coparent{j}}_{1:t-1}) 
\end{multline}
By the same argument as in part 2, the last term is independent of $\vec{i}^{\anc{j}}$ and therefore its estimate is $0$ for any choice of $\vec{m}^{\sdes{j}}$. The rest of the proof is the same as in part 2.

\item To prove Part 4, we make the following observation about conditional means: Consider jointly Gaussian random vectors $A$ and $B$. Define linear combinations of these as follows. 
\[ C:= \mathcal{K}(A),  ~~~~~~ D:= C + \mathcal{L}(B) ,\]
where $\mathcal{K,L}$ are linear functions. Then, with probability $1$,
\[ C -\eec{C}{B} = D- \eec{D}{B}\]
and hence the conditional distribution of the random vector $(C -\eec{C}{B})$ given $B=b$ is the same as the conditional distribution of the random vector $(D- \eec{D}{B})$ given $B=b$. In other words,  when the difference between $C$ and $D$ is measurable with respect to $B$, the conditional distribution of the estimation error  of $C$ given $B$ and $D$ given $B$ are the same.


In Part 4, the difference between the strategy profiles $f$ and $g$ is due to the functions $\ell^i_t(\vec{i}^{\anc{j}}_t)$ in the strategies of descendants of $j$. Because this function is measurable with respect to $j$'s information, node $j$ can account for its effect without changing the error. More precisely, under the new strategy $g$, we can write 
\begin{multline} \label{eq:est_graph5}
\vec{u}^{\sdes{j}}_{1:t-1}  \in \lin^{f^\sdes{j}}(\vec{x}^{\funnel{j}}_0,\vec{w}^{\funnel{j}}_{1:t-1},\vec{v}^{\funnel{j}}_{1:t-1}, \vec{u}^{\anc{j}}_{1:t-1},\vec{x}^{\coparent{j}}_0,\vec{w}^{\coparent{j}}_{1:t-1},\vec{v}^{\coparent{j}}_{1:t-1}, \vec{u}^{\coparent{j}}_{1:t-1}) 
+ \lin^{{\ell}^{\sdes{j}}}(\vec{i}^{\anc{j}}_t) 
\end{multline}
and use it with with \eqref{eq:est_graph0} and the observation equation from \eqref{eq:state_eqns_sum} to get that under strategy $g$
\begin{multline}\label{eq:est_graph6}
 (\vec{x}^{\funnel{j}}_t,\vec{i}^{\funnel{j}}_t,\vec{y}^{\anc{j}}_t )\in \lin^{f^\sdes{j}}(\vec{x}^{\funnel{j}}_0,\vec{w}^{\funnel{j}}_{1:t-1},\vec{v}^{\funnel{j}}_{1:t}, \vec{u}^{\anc{j}}_{1:t-1} )+ \lin^{f^\sdes{j}}(\vec{x}^{\coparent{j}}_0,\vec{w}^{\coparent{j}}_{1:t-1},\vec{v}^{\coparent{j}}_{1:t-1}, \vec{u}^{\coparent{j}}_{1:t-1}) 
 \\+ \lin^{{\ell}^{\sdes{j}}}(\vec{i}^{\anc{j}}_t) 
\end{multline}
Under the original strategy $f$, \eqref{eq:est_graph6} would have been 
\begin{multline}\label{eq:est_graph7}
 (\vec{x}^{*,\funnel{j}}_t,\vec{i}^{*,\funnel{j}}_t,\vec{y}^{*,\anc{j}}_t )\in \lin^{f^\sdes{j}}(\vec{x}^{\funnel{j}}_0,\vec{w}^{\funnel{j}}_{1:t-1},\vec{v}^{\funnel{j}}_{1:t}, \vec{u}^{\anc{j}}_{1:t-1} )+ \lin^{f^\sdes{j}}(\vec{x}^{\coparent{j}}_0,\vec{w}^{\coparent{j}}_{1:t-1},\vec{v}^{\coparent{j}}_{1:t-1}, \vec{u}^{\coparent{j}}_{1:t-1}),
 \end{multline}
 where we use $*$ to indicate that these random variables are defined under a different strategy than those in \eqref{eq:est_graph6}.
 
Then, using the argument above with $\vec{i}^{\anc{j}}_t$ playing the role of random vector $B$ and the right hand sides of \eqref{eq:est_graph7} and \eqref{eq:est_graph6} playing the roles of $C$ and $D$ respectively, we can state that the  conditional distribution  of the vector 
\[ \left((\vec{x}^{*,\funnel{j}}_t,\vec{i}^{*,\funnel{j}}_t, \vec{y}^{*,\anc{j}}_t )) - \ee^{f}\left(\vec{x}^{*,\funnel{j}}_t,\vec{i}^{*,\funnel{j}}_t, \vec{y}^{*,\anc{j}}_t\middle\vert\vec i^{\anc{j}}_t = i^{\anc{j}}_t\right)\right)\]
given $i^{\anc{j}}_t$ when the strategy profile is $f$
is the same as the the  conditional distribution  of the vector 
\[ \left((\vec{x}^{\funnel{j}}_t,\vec{i}^{\funnel{j}}_t, \vec{y}^{\anc{j}}_t )) - \ee^{g}\left(\vec{x}^{\funnel{j}}_t,\vec{i}^{\funnel{j}}_t, \vec{y}^{\anc{j}}_t\middle\vert\vec i^{\anc{j}}_t = i^{\anc{j}}_t\right)\right)\]
given $i^{\anc{j}}_t$ when the strategy profile is $g$. Consequently, the corresponding conditional covariance matrices are the same.

\item Because of Part 4, we can remove the functions $\ell^i_t(\vec{i}^{\anc{j}}_t)$ while computing conditional covariances.  

The random vector $\vec{y}^{\anc{j}}_t, \vec{i}^{\anc{j}}_t$ can be written in terms of primitive random variables and control decisions as
\[(\vec{y}^{\anc{j}}_t,\vec{i}^{\anc{j}}_t) \in \lin(\vec{x}^{\anc{j}}_0,\vec{w}^{\anc{j}}_{1:t-1},\vec{v}^{\anc{j}}_{1:t}, \vec{u}^{\anc{j}}_{1:t-1}), \]
where the $\lin$ function depends only on system parameters. Using arguments similar to those used to get  \eqref{eq:est_graph1b}, \eqref{eq:est_graph2}, it can be seen that the conditional distribution of $\vec{y}^{\anc{j}}_t$ given $\vec{i}^{\anc{j}}_t$ does not depend on the strategy profile.
 Therefore, the conditional covariance of $\vec{y}^{\anc{j}}_t$ given $\vec{i}^{\anc{j}}_t$ will not depend on the choice of $m^i_t(\vec{i}^{\coparent{j}}_t), i \in \sdes{j}$. 

For the cross-covariance matrix between $(\vec{x}^{\funnel{j}}_t,\vec{i}^{\funnel{j}}_t)$ and  $\vec{y}^{\anc{j}}_t $ given $\vec{i}^{\anc{j}}_t$, we use \eqref{eq:est_graph4} for computing $(\vec{x}^{\funnel{j}}_t,\vec{i}^{\funnel{j}}_t)\vec{y}^{\anc{j}, \tp}_t$ and once again use the independence of $\vec{i}^{\coparent{j}}_t$ and $\vec{y}^{\anc{j}}_t,\vec{i}^{\anc{j}}_t$ to remove the terms that depend on the functions $m^i_t(\vec{i}^{\coparent{j}}_t)$.
\end{enumerate}

\section{Proof of Lemma \ref{lem:leaf_sixnode}}\label{sec:leaf_sixnode}
For a leaf node $j$, since $\sdes{j} = \coparent{j} = \emptyset$, the aggregated graph centered at $j$ looks like Figure~\ref{fig:leaf_aggregate}. If strategies of all nodes except node $j$ are fixed, then controller $j$'s optimization problem is a centralized problem. We will now show that with the state, measurement and control definitions of \eqref{eq:leaf_to_6_node}, the system dynamics, the cost and the information structure of controller $j$'s problem are identical to the centralized Problem \ref{prob:CLQG_6node}.
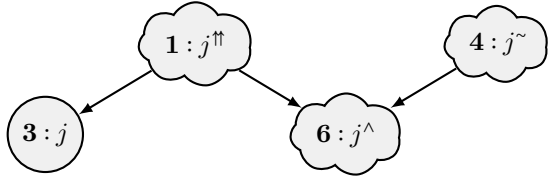
\begin{figure}[ht]
\centering
\begin{tikzpicture}[thick,>=latex]
	\tikzstyle{block}=[circle,draw,fill=black!6,minimum height=2.4em];
	\tikzstyle{cblock}=[cloud,draw,fill=black!6,cloud puffs=7.5,cloud puff arc=120, aspect=1.7, inner sep=1mm];
	\def\x{2.0};
	\def\y{1.2};
	\def\t{0.5};
	\node [cblock](N1) at (-\t*\x,\t*\y) {$\textbf{1}:\sanc{j}$};
	\node [block](N3) at (-3*\t*\x,-\t*\y) {$\textbf{3}:j$};
	\node [cblock](N4) at (3*\t*\x,\t*\y) {$\textbf{4}:\nonrelative{j}\!$};
    \node [cblock](N6) at (\t*\x,-\t*\y) {$\textbf{6}:\sibling{j}$};
	\draw [->] (N1) -- (N3);
	\draw [->] (N1) -- (N6);
	\draw [->] (N4) -- (N6);
\end{tikzpicture}
\caption{The aggregated graph centered at leaf node ~$j$.\label{fig:leaf_aggregate}}
\end{figure}

\noindent Given a node $i\in\mathcal{V}$, the state equations~\eqref{eq:state_eqns_sum} and information structure~\eqref{eq:information_pattern} imply that
\begin{align}\label{eq:base_rec}
\vec{x}_{t+1}^i &\in\lin(\vec{x}_t^\anc{i},\vec{u}_t^\anc{i},\vec w_t^i)
&
\vec{y}_t^i &\in\lin(\vec{x}_t^\anc{i},\vec v_t^i)
&
\vec{i}_{t+1}^i &\in\lin( \vec{i}_t^i, \vec{y}_t^i, \vec{u}_t^i )
\end{align}
Given the fixed strategies for nodes $i\neq j$, the input definition~\eqref{eq:input_definition} implies that 
$\vec{u}_t^i \in\lin(\vec{i}_t^\anc{i})$ for all $i\ne j$ which further implies that $\vec{u}_t^{\anc{i}} \in\lin(\vec{i}_t^\anc{i})$.  Combining these facts, we deduce that
\begin{align}\label{eq:rec_aug_lin_state}
\bmat{ \vec{x}_{t+1}^i \\ \vec{i}_{t+1}^i }
\in\lin\left( \bmat{ \vec{x}_{t}^\anc{i} \\ \vec{i}_{t}^\anc{i} }, \bmat{\vec{w}_t^i \\ \vec{v}_t^i} \right)
\qquad\text{for $i\ne j$}
\end{align}
We now  apply~\eqref{eq:rec_aug_lin_state}  to different possible subsets of $\mathcal{V}$ as identified in Figure~\ref{fig:leaf_aggregate}. For example, since $\sanc{j}$ and $\nonrelative{j}$ are equal to their respective ancestral sets,
\begin{align}\label{eq:rec_aug_a}
\bar{\vec{x}}^1_{t+1}
&\in\lin\left( \bar{\vec{x}}^1_{t}, \bmat{\vec{w}_t^\sanc{j} \\ \vec{v}_t^\sanc{j}} \right)
& &\text{and} &
\bar{\vec{x}}^4_{t+1}
&\in\lin\left( \bar{\vec{x}}^4_t, \bmat{\vec{w}_t^\nonrelative{j} \\ \vec{v}_t^\nonrelative{j}} \right)
\end{align}
For the sibling set, if $k\in \sibling{j}$, then $\anc{k} \subseteq \sanc{j}\cup \sibling{j}\cup \nonrelative{j}$. Therefore we obtain
\begin{align}\label{eq:rec_aug_b}
\bar{\vec{x}}^6_{t+1} &\in\lin\left( \bar{\vec{x}}^1_{t}, \bar{\vec{x}}^4_{t}, \bar{\vec{x}}^6_{t}, \bmat{\vec{w}_t^\sibling{j} \\ \vec{v}_t^\sibling{j}} \right)
\end{align}
For $\vec{x}^j_t$, we have
$\vec{x}^j_{t+1} \in\lin(\vec{x}_t^\anc{j},\vec{u}_t^\anc{j},\vec w_t^j)$.
Splitting $\anc{j} = \{j\}\cup\sanc{j}$, we obtain
\begin{align}\label{eq:rec_aug_c}
\bar{\vec{x}}^3_{t+1} &\in\lin\left( \bar{\vec{x}}^1_{t}, \bar{\vec{x}}^3_t, \vec{u}_t^j, \vec{w}_t^j \right)
\end{align}
A similar argument 
can be used to conclude that
\begin{align}\label{eq:rec_aug_d}
\bar{\vec{y}}^3_t
&\in \lin( \bar{\vec{x}}^1_t, \bar{\vec{x}}^3_t, \vec{v}_t^{\anc{j}} )
\end{align}
It follows from~\eqref{eq:rec_aug_a}--\eqref{eq:rec_aug_d} that the dynamics of controller $j$'s problem have the same structure as the dynamics  of Problem \ref{prob:CLQG_6node} given in~\eqref{eq:six_node_eqns_structure}. 

It is clear that the information available to controller $j$ can be written as $\Bigl\{ \bar{\vec{y}}_{0:t-1}^{3}, \bar{\vec{u}}_{0:t-1}^{3} \Bigr\}$. Further, the cost and covariance matrices have the sparsity structure required in Problem~\ref{prob:CLQG_6node}.

\section{Proof of Lemma \ref{lem:Gs_lemma}} \label{sec:Gs_lemma}
We will show that with the state, measurement and control definitions of \eqref{eq:newcoord_to_6_node}, the system dynamics, the cost and the information structure of the coordinator's problem are identical to the centralized Problem \ref{prob:CLQG_6node}.
As in the proof of Lemma~\ref{lem:leaf_sixnode}, \eqref{eq:base_rec}  still holds. 
 Following an argument analogous to the one used in proving Lemma~\ref{lem:leaf_sixnode}, it follows that 
\begin{align}
\bar{\vec{x}}^1_{t+1}
&\in\lin\left( \bar{\vec{x}}^1_{t}, \bmat{\vec{w}_t^\sanc{k} \\ \vec{v}_t^\sanc{k}} \right)
& &,&
\bar{\vec{x}}^2_{t+1}
&\in\lin\left( \bar{\vec{x}}^2_t, \bmat{\vec{w}_t^\coparent{k} \\ \vec{v}_t^\coparent{k}} \right)\\
\bar{\vec{x}}^4_{t+1} &\in\lin\left( \bar{\vec{x}}^2_{t}, \bar{\vec{x}}^4_{t}, \bmat{\vec{w}_t^\nonrelative{k} \\ \vec{v}_t^\nonrelative{k}} \right) & &, &
\bar{\vec{x}}^6_{t+1} &\in\lin\left( \bar{\vec{x}}^1_{t},\bar{\vec{x}}^2_{t}, \bar{\vec{x}}^4_{t}, \bar{\vec{x}}^6_{t}, \bmat{\vec{w}_t^\sibling{k} \\ \vec{v}_t^\sibling{k}} \right)\\
\bar{\vec{x}}^3_{t+1} &\in\lin\left( \bar{\vec{x}}^1_{t}, \bar{\vec{x}}^3_t, \vec{u}_t^k, \vec{w}_t^k \right) & &\text{and} &
\bar{\vec{y}}^3_t&\in \lin( \bar{\vec{x}}^1_t, \bar{\vec{x}}^3_t, \vec{v}_t^{\anc{k}} )
\end{align}
 Consider $j\in\sdes{k}$. We have that 
 \begin{equation}
 \vec{x}^j_{t+1} \in
	\lin( \vec{x}^\anc{j}_t, \vec{u}^\anc{j}_t, \vec{w}^j_t ) \label{eq:lemma6_1}
 \end{equation}
  Note that $\anc{j} \subseteq \sdes{k} \cup \{k\}\cup \sanc{k} \cup  \coparent{k}$. Therefore, \eqref{eq:lemma6_1} can be written as 
\begin{align}
&\vec{x}^j_{t+1} \in
	\lin( \vec{x}^\anc{j}_t, \vec{u}^\anc{j}_t, \vec{w}^j_t )  \notag\\
&\in\lin( \vec{x}^\sdes{k}_t, \vec{u}^\sdes{k}_t,  \vec{x}^k_t, \vec{u}^k_t,
	\vec{x}^\sanc{k}_t, \vec{u}^\sanc{k}_t, \vec{x}^\coparent{k}_t, \vec{u}^\coparent{k}_t, \vec{w}^j_t ) \label{eq:lemma6_2}
\end{align}
Using \eqref{turtle} for $\vec{u}^{\sdes{k}}_t$ and the fact that for $c \in \coparent{k} \cap \mathcal{G}^{\leq s}$,  $\vec{z}^{\funnel{c}}_t \in \lin(\vec{i}^{\coparent{k}}_t)$  in \eqref{eq:lemma6_2}, we get
\begin{align}
\vec{x}^j_{t+1} &\in \lin( \vec{x}^\sdes{k}_t,  \{\vec z^\funnel{i}_t, \vec{\tilde u}_t^{ik}\}_{i \in \sdes{k}},  \vec i^\coparent{k}_t, \vec{x}^k_t, \vec{u}^k_t,
	\vec{x}^\sanc{k}_t, \vec{u}^\sanc{k}_t, \vec{x}^\coparent{k}_t, \vec{u}^\coparent{k}_t, \vec{w}^j_t ) \notag \\
&\in \lin( \vec{x}^\sdes{k}_t,  \{\vec z^\funnel{i}_t, \vec{\tilde u}_t^{ik}\}_{i \in \sdes{k}},  \vec i^\coparent{k}_t, \vec{x}^k_t, \vec{u}^k_t,
	\vec{x}^\sanc{k}_t, \vec{i}^\sanc{k}_t, \vec{x}^\coparent{k}_t, \vec{w}^j_t ) \label{eq:lemma6_3}
\end{align}
where we used the fact that given the linear strategy of $\sanc{k}, \coparent{k}$, $\vec{u}^{\sanc{k}}_t \in \lin(\vec{i}^{\sanc{k}}_t)$ and $\vec{u}^{\coparent{k}}_t \in \lin(\vec{i}^{\coparent{k}}_t)$. Grouping the various terms in \eqref{eq:lemma6_3} gives
\begin{align}
\vec{x}^j_{t+1} \in\lin(\bar{\vec{x}}^1_t, \bar{\vec{x}}^2_t, \bar{\vec{x}}^3_t,
	\bar{\vec{x}}^5_t, \bar{\vec{u}}^3_t,\vec w_t^j) \label{eq:rec_aug_lin_state2A}
\end{align}
Now we turn to the dynamics of $\vec z^\funnel{j}$ for $j\in\sdes{k}$. From \eqref{eq:z_Gs_dynamics} we have
\begin{equation}
\vec{z}_{t+1}^\funnel{j} \in \lin( \vec{z}_t^\funnel{j}, \vec{u}_t^\anc{j},  \{\vec{\hat u}_t^{ij}\}_{i \in \sdes{j}}, \vec{y}_t^\anc{j} ) 
\label{eq:lemma6_4}
\end{equation}
We will focus on the terms $\vec{u}_t^\anc{j},  \{\vec{\hat u}_t^{ij}\}_{i \in \sdes{j}}, \vec{y}_t^\anc{j}$ appearing in \eqref{eq:lemma6_4}.
\begin{itemize}
\item $\vec{u}_t^\anc{j}$: Using arguments similar to those used to get \eqref{eq:lemma6_3}, we get 
\begin{equation}
\vec{u}_t^\anc{j} \in \lin(\{\vec z^\funnel{i}_t, \vec{\tilde u}_t^{ik}\}_{i \in \sdes{k}},  \vec i^\coparent{k}_t, \vec i^\sanc{k}_t, \vec{u}^k_t) \label{eq:lemma6_5}
\end{equation}
\item $\{\vec{\hat u}_t^{ij}\}_{i \in \sdes{j}}$: Using \eqref{eq:hat_tilde}, we get
\begin{equation}\label{eq:lemma6_6}
\{\vec{\hat u}_t^{ij}\} \in \lin(\{\vec z^\funnel{a}_t, \vec{\tilde u}_t^{ak}\}_{a \in \sdes{k}},  \vec i^\coparent{k}_t)
\end{equation}
\item $\vec{y}_t^\anc{j}$: From the system model, we can write
\begin{equation} \label{eq:lemma6_7}
\vec{y}_t^\anc{j} \in \lin(\vec{x}^{\sanc{k}}_t,\vec{x}^{{k}}_t,\vec{x}^{\sdes{k}}_t,\vec{x}^{\coparent{k}}_t \vec{v}^{\anc{j}}_t)
\end{equation}
\end{itemize}
Combining \eqref{eq:lemma6_4}-\eqref{eq:lemma6_7}, we get
\begin{align}\label{eq:rec_aug_lin_state3A}
\vec{z}_{t+1}^\funnel{j} &\in \lin( \vec{z}_t^\funnel{j}, \{\vec z^\funnel{i}_t, \vec{\tilde u}_t^{ik}\}_{i \in \sdes{k}},  \vec i^\coparent{k}_t, \vec i^\sanc{k}_t, \vec{u}^k_t,\vec{x}^{\sanc{k}}_t,\vec{x}^{{k}}_t,\vec{x}^{\sdes{k}}_t,\vec{x}^{\coparent{k}}_t \vec{v}^{\anc{j}}_t)
\end{align}
Grouping the various terms in \eqref{eq:rec_aug_lin_state3A}, we get
\begin{align}
\vec{z}_{t+1}^\funnel{j}\in\lin(\bar{\vec{x}}^1_t, \bar{\vec{x}}^2_t, \bar{\vec{x}}^3_t,
	\bar{\vec{x}}^5_t, \bar{\vec{u}}^3_t,\vec v_t^\anc{j}) \label{eq:rec_aug_lin_state4}
\end{align}
Combining \eqref{eq:rec_aug_lin_state2A} and \eqref{eq:rec_aug_lin_state4}, we conclude that $\bar{\vec x}_{t+1}^5 \in\lin(\bar{\vec{x}}^1_t, \bar{\vec{x}}^2_t, \bar{\vec{x}}^3_t,
	\bar{\vec{x}}^5_t, \bar{\vec{u}}^3_t,\hat{\vec{w}}_t)$, as required, where $\hat{\vec w}_t \in\lin(\vec w_t^j,\vec v_t^\anc{j}), j \in \sdes{k}$ is an aggregated noise term. Therefore, the dynamics of coordinator's problem are identical to the dynamics  of Problem \ref{prob:CLQG_6node} given in~\eqref{eq:six_node_eqns_structure}.  It is clear that the information available to the coordinator at node $j$ can be written as $\Bigl\{ \bar{\vec{y}}_{0:t-1}^{3}, \bar{\vec{u}}_{0:t-1}^{3} \Bigr\}$. Further, the cost and covariance matrices have the sparsity structure required in Problem~\ref{prob:CLQG_6node}.
\bibliographystyle{abbrv}
\bibliography{pn}
	\end{document}